\documentclass[a4paper]{article}

\usepackage{a4wide,graphicx}
\usepackage{amsmath,amssymb,amsthm}
\usepackage{mathrsfs,stmaryrd,bbm}

\newcommand{\abs}[1]{\vert #1\vert}
\newcommand{\ave}[1]{\langle #1\rangle}
\newcommand{\dual}[2]{\langle #1,\,#2\rangle}
\newcommand{\innprod}[2]{(#1,\,#2)}
\newcommand{\jump}[1]{\llbracket #1\rrbracket}
\newcommand{\intl}{\int\limits}
\newcommand{\Lip}[1]{\textnormal{Lip}(#1)}
\newcommand{\LipPhi}[1]{\textnormal{Lip}_{\Phi}(#1)}
\newcommand{\norm}[1]{\|#1\|}
\newcommand{\winnprod}[2]{(\!(#1,\,#2)\!)}
\newcommand{\wnorm}[1]{|\!|\!|#1|\!|\!|}

\newcommand{\1}{\mathbbm{1}}
\newcommand{\A}{\mathcal{A}}
\newcommand{\B}{\mathcal{B}}
\newcommand{\cth}{c_\ast}
\newcommand{\D}{\mathbf{D}}
\newcommand{\Heav}{H}
\newcommand{\I}{\mathscr{I}}
\newcommand{\K}{\mathbf{K}}
\newcommand{\Lag}{L}
\newcommand{\m}{\mathbf{m}}
\newcommand{\n}{\mathbf{n}}
\renewcommand{\P}{\mathscr{P}}
\newcommand{\QT}{\mathcal{Q}_{\Tmax}}
\newcommand{\R}{\mathbb{R}}
\renewcommand{\S}{\mathscr{S}}
\newcommand{\T}{\mathbf{T}}
\newcommand{\Tmax}{T_\textup{max}}
\newcommand{\UU}{\mathcal{U}}
\renewcommand{\v}{\mathbf{v}}
\newcommand{\V}{\mathbb{V}}
\newcommand{\VT}{\mathbb{V}_{\Tmax}}
\newcommand{\VV}{\mathcal{V}}

\newcommand{\phimax}{\phi_\textup{max}}
\newcommand{\phith}{\phi_\ast}

\newtheorem{theorem}{Theorem}
\theoremstyle{definition}\newtheorem{definition}{Definition}
\theoremstyle{remark}\newtheorem{remark}{Remark}

\graphicspath{{./}{figures/}}

\title{Initial/boundary-value problems of tumor growth within a host tissue}
\author{Andrea Tosin\thanks{A. Tosin was funded by a post-doctoral research scholarship ``Compagnia di San Paolo''
							awarded by the National Institute for Advanced Mathematics ``F. Severi'' (INdAM, Italy).} \\[5mm]
		{\small\it Department of Mathematics, Politecnico di Torino}\\[-1mm]
		{\small\it Corso Duca degli Abruzzi 24, 10129, Torino, Italy}
	   }
\date{}

\begin{document}

\maketitle

\begin{abstract}
This paper concerns multiphase models of tumor growth in interaction with a surrounding tissue, taking into account also the interplay with diffusible nutrients feeding the cells. Models specialize in nonlinear systems of possibly degenerate parabolic equations, which include phenomenological terms related to specific cell functions. The paper discusses general modeling guidelines for such terms, as well as for initial and boundary conditions, aiming at both biological consistency and mathematical robustness of the resulting problems. Particularly, it addresses some qualitative properties such as \emph{a priori} nonnegativity, boundedness, and uniqueness of the solutions. Existence of the solutions is studied in the one-dimensional time-independent case.

\medskip


\medskip

\end{abstract}

\section{Mixture-theory equations for tumor growth}
\subsection{Mixture-theory-based models}
The interest toward mathematical modeling of tumor growth rose considerably in the last decades, to such an extent that it has now become one of the most studied topics in mathematical biology. Early mathematical models \cite{MR2005055,byrne1995gnt,byrne1996gnt} considered tumors as ensembles of only one type of cells. Growth was described under the main assumption of constant cell density, by relating the volume variation of the tumor mass to birth and death of cells triggered by nutrient supply. In most cases simple \emph{in vitro} geometries were considered, such as spheroids, and qualitative analyses of the resulting free boundary problems were detailed \cite{MR2150346,bueno2008ssm,chen2003fbp,MR1815805,friedman2009fbp,friedman2007bfb,MR1684873}.

However, the biological literature pointed out soon that tumors should be regarded more properly as ensembles of different interacting components, e.g., normal and abnormal cells, intercellular fluid, extracellular matrix. This aspect is taken into account by modeling tumors as  multiphase materials by methods of mixture theory, see for instance \cite{byrne2003mst}. In mixture theory one introduces a few volume ratios $\phi_\alpha$, where the index $\alpha$ labels the components of the mixture, expressing the percent amount of the constituents. Each volume ratio is supposed to satisfy $0\leq\phi_\alpha\leq 1$, with a possible further condition $\sum_\alpha\phi_\alpha=1$, called \emph{saturation constraint}, if one assumes that no voids are left within the mixture. Mass balance equations are written for the constituents under the assumption of same density:
\begin{equation}
	\frac{\partial\phi_\alpha}{\partial t}+\nabla\cdot(\phi_\alpha\v_\alpha)=\Gamma_\alpha,
	\label{eq:mixture-mass}
\end{equation}
where $\v_\alpha$ and $\Gamma_\alpha$ are the velocity and the source/sink term of the constituent $\alpha$, respectively.

A first class of multiphase models is obtained from Eq.~\ref{eq:mixture-mass} via suitable closure relations relating the velocities to the volume ratios of the constituents \cite{MR1909425}. A very common assumption is that all cell populations share the same velocity, while non-cellular components have their own. Then geometrical considerations, for instance some symmetries, may determine kinematically the velocities of some constituents, as it happens in \cite{MR2111919} where the cylindrical symmetry of tumor cords developing along a blood vessel is used to deduce the velocity of the cells. The model presented in that paper combines ideas coming from mixture theory with free boundary issues. Particularly, it includes dynamical constraints on the nutrient distribution across the tumor mass for modeling the formation of an outer necrotic shell of dead cells at the periphery of the tumor cord.

A second class of multiphase models is obtained by joining to Eq.~\ref{eq:mixture-mass} some stress balance equations for the constituents of the mixture, in which inertial effects are neglected:
\begin{equation}
	-\nabla\cdot(\phi_\alpha\T_\alpha)+\phi_\alpha\nabla{p}=\m_\alpha.
	\label{eq:mixture-stress}
\end{equation}
Here $\T_\alpha$, $\m_\alpha$ are the excess stress tensor and the resultant of the external actions on the constituent $\alpha$, respectively, whereas $p$ is the intercellular fluid pressure. Equations~\ref{eq:mixture-stress} are used to derive the velocities $\v_\alpha$ from mechanical reasonings on the internal and external stress sustained by the constituents. For instance, assuming that the external actions $\m_\alpha$ can be expressed as viscous frictions, thus involving only the relative velocity of pairs of constituents, one gets (possibly generalized) Darcy's laws that can be plugged into Eq.~\ref{eq:mixture-mass}. Considering in particular a sub-mixture of two cell populations, namely tumor cells labeled with $\alpha=T$ and healthy host cells labeled with $\alpha=H$, which share the same mechanical properties, and assuming that the extracellular matrix only acts as a rigid non-remodeling scaffold providing them with a support for their movement, the following equations are obtained \cite{MR2471305}:
\begin{equation}
	\frac{\partial\phi_\alpha}{\partial t}-\nabla\cdot{\left[\frac{\phi_\alpha^2}{\phi}
		\K_{\alpha m}\nabla{(\phi\Sigma(\phi))}\right]}=\Gamma_\alpha,
	\label{eq:mixture-general}
\end{equation}
where $\phi:=\phi_T+\phi_H$ is the overall volume ratio of the cellular matter, $\Sigma=\Sigma(\phi)$ is the intercellular stress (such that $\T_T=\T_H=-\Sigma(\phi)\mathbf{I}$), and finally $\K_{\alpha m}$ is the motility tensor of the cell population $\alpha$ within the extracellular matrix.

The source/sink terms $\Gamma_\alpha$ model proliferation or death of cells, taking into account both natural processes linked to the vital cell cycle and the availability of nutrient. Hence they depend on the cell volume ratios $\phi_T,\,\phi_H$ and on the nutrient concentration $c$, which entails $\Gamma_\alpha=\Gamma_\alpha(\phi_T,\,\phi_H,\,c)$. This introduces the new variable $c$, for which an evolution equation, usually of reaction-diffusion type, is supplied:
\begin{equation}
	\frac{\partial c}{\partial t}-\nabla\cdot{(\D\nabla{c})}=\sum_{\alpha=T,\,H}Q_\alpha(\phi_\alpha,\,c),
	\label{eq:mixture-c}
\end{equation}
where $\D$ is the diffusivity tensor and $Q_\alpha$ models the absorption of nutrient by the cells of the population $\alpha$. The tensor $\D$ may be taken independent of $c$, and also of the cell volume ratios $\phi_\alpha$, because nutrient molecules are not regarded as a part of the mixture, i.e., they are assumed to diffuse through the mixture without occupying space. Equation~\ref{eq:mixture-c} features then a linear diffusion. In addition, the functions $Q_\alpha$ are sometimes linear in $c$, hence Eq.~\ref{eq:mixture-c} turns out to be, in most cases, a linear model for the nutrient concentration. However, the simultaneous dependence of the functions $\Gamma_\alpha$, $Q_\alpha$ on both $\phi_\alpha$, $c$ makes Eqs.~\ref{eq:mixture-general}, \ref{eq:mixture-c} ultimately coupled.

Technical simplifications in Eqs.~\ref{eq:mixture-general}, \ref{eq:mixture-c} may involve the assumption of homogeneous isotropic motility of the cells in the extracellular matrix, as well as homogeneous isotropic diffusion of the nutrient through the mixture. These imply $\K_{\alpha m}=\kappa_{\alpha m}\mathbf{I}$ in Eq.~\ref{eq:mixture-general} and $\D=D\mathbf{I}$ in Eq.~\ref{eq:mixture-c}, for positive constants $\kappa_{\alpha m},\,D$. In addition, invoking the hypothesis of same mechanical properties for tumor and host cells, one may set $\kappa_{Tm}=\kappa_{Hm}=:\kappa_m>0$, which, as pointed out in \cite{MR2471305}, is a good approximation at least in the initial stages of the development of a tumor, when contact inhibition among cells is more influential than differences in their motility.

\subsection{Cell segregation}

\begin{figure*}
\centering
\includegraphics[width=0.5\textwidth]{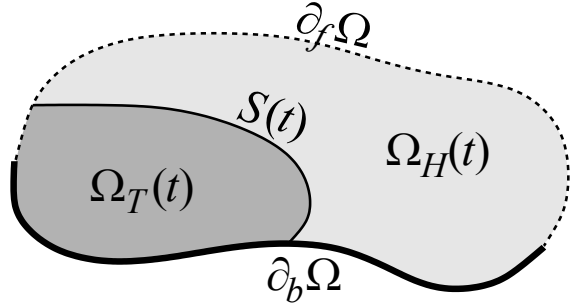}
\caption{The domain $\Omega$ in case of cell segregation. At each time, tumor cells are contained in $\Omega_T(t)$ and healthy host cells in $\Omega_H(t)$. The two cell populations are separated by the interface $S(t)$, which is a $(d-1)$-dimensional manifold in $\Omega$.}
\label{fig:domain}
\end{figure*}

An interesting case is when tumor and host cells remain segregated, i.e., the spatial domain $\Omega\subset\R^d$ of the problem can be split at each time in two (open) sub-domains $\Omega_T(t)$, $\Omega_H(t)$ such that $\overline{\Omega_T(t)\cup\Omega_H(t)}=\overline{\Omega}$ and $\Omega_T(t)\cap\Omega_H(t)=\emptyset$ (Fig. \ref{fig:domain}). In practice, each sub-domain contains a mixture of extracellular fluid, extracellular matrix, and just one type of cells obeying the following balance equation:
\begin{equation}
	\frac{\partial\phi_\alpha}{\partial t}-\kappa_m\nabla\cdot{\left[\phi_\alpha
		\nabla{(\phi_\alpha\Sigma(\phi_\alpha))}\right]}=\Gamma_\alpha(\phi_\alpha,\,c),
	\label{eq:mixture}
\end{equation}
each volume ratio $\phi_\alpha$ being defined only in the corresponding domain $\Omega_\alpha$. Equation~\ref{eq:mixture} is derived from Eq.~\ref{eq:mixture-general} noticing that $\phi\equiv\phi_\alpha$ in $\Omega_\alpha$ owing to segregation. The two mixtures interact at the interface $S(t):=\partial\Omega_T(t)\cap\partial\Omega_H(t)$ separating the sub-domains, hence each $\phi_\alpha$ solves in principle a free boundary problem because $S(t)$ is not fixed. However, it is possible to supplement Eq.~\ref{eq:mixture} with proper conditions on $S(t)$ so as to reformulate it globally in $\Omega$. Proceeding in a formal fashion, we integrate Eq.~\ref{eq:mixture} on $\Omega_\alpha(t)$ up to a certain final time $\Tmax>0$, then we apply Gauss' Theorem to the divergence term at the left-hand side and sum over $\alpha$ to discover:
\begin{align*}
	\sum_{\alpha=T,\,H}&\left(\intl_0^{\Tmax}
		\intl_{\Omega_\alpha(t)}\frac{\partial\phi_\alpha}{\partial t}\,dx\,dt-
		\kappa_m\intl_0^{\Tmax}\intl_{\partial\Omega_\alpha(t)\setminus S(t)}\phi_\alpha
			\nabla{(\phi_\alpha\Sigma(\phi_\alpha))}\cdot\n\,d\sigma\,dt\right) \\
	&+\kappa_m\intl_0^{\Tmax}\intl_{S(t)}\left[\phi_T\nabla{(\phi_T\Sigma(\phi_T))}-
		\phi_H\nabla{(\phi_H\Sigma(\phi_H))}\right]\cdot\n\,d\sigma\,dt \\
	&=\sum_{\alpha=T,\,H}
		\intl_0^{\Tmax}\intl_{\Omega_\alpha(t)}\Gamma_\alpha(\phi_\alpha,\,c)\,dx\,dt,
\end{align*}
where $d\sigma$ is the $(d-1)$-dimensional Hausdorff measure in $\R^d$ and $\n$ is the outward normal unit vector to the boundary on which integration is performed. In particular, along $S(t)$ it denotes the outward normal unit vector to $\Omega_H(t)$, so that the analogous vector for $\Omega_T(t)$ is $-\n$ (but, of course, the opposite convention may also be adopted).

Next we reintroduce the function $\phi:[0,\,\Tmax]\times\Omega\to\R$:
\begin{equation*}
	\phi(t,\,x)=
	\begin{cases}
		\phi_T(t,\,x) & \text{if\ } x\in\Omega_T(t) \\
		\phi_H(t,\,x) & \text{if\ } x\in\Omega_H(t),
	\end{cases}
	\quad t\in[0,\,\Tmax],
\end{equation*}
and use $\cup_\alpha\Omega_\alpha(t)=\Omega$, $\cup_\alpha\partial\Omega_\alpha(t)\setminus S(t)=\partial\Omega$ to obtain
\begin{multline}
	\intl_0^{\Tmax}\intl_\Omega\frac{\partial\phi}{\partial t}\,dx\,dt-
	\kappa_m\intl_0^{\Tmax}\intl_{\partial\Omega}\phi\nabla{(\phi\Sigma(\phi))}\cdot\n\,d\sigma\,dt \\
	+\kappa_m\intl_0^{\Tmax}\intl_{S(t)}\jump{\phi\nabla{(\phi\Sigma(\phi))}}\cdot\n\,d\sigma\,dt
	= \intl_0^{\Tmax}\intl_\Omega\Gamma(t,\,x,\,\phi,\,c)\,dx\,dt,
	\label{eq:mixture-segr-int}
\end{multline}
where we have defined 
\begin{equation}
	\Gamma(t,\,x,\,\phi,\,c):=\sum_{\alpha=T,\,H}\Gamma_\alpha(\phi_\alpha,\,c)\1_{\Omega_\alpha(t)}(x).
	\label{eq:Gamma_on_Om}
\end{equation}
In formulas \ref{eq:mixture-segr-int}, \ref{eq:Gamma_on_Om}, $\jump{\cdot}$ denotes jump across $S(t)$ whereas $\1_{\Omega_\alpha(t)}$ is the indicator function of the set $\Omega_\alpha(t)$. If we reapply Gauss' Theorem to the second term at the left-hand side, we can regard Eq.~\ref{eq:mixture-segr-int} as the integral version of the differential equation
\begin{equation}
	\frac{\partial\phi}{\partial t}-\kappa_m\nabla\cdot{\left[\phi
		\nabla{(\phi\Sigma(\phi))}\right]}=\Gamma(t,\,x,\,\phi,\,c)
	\label{eq:system-phi}
\end{equation}
provided $\kappa_m\jump{\phi\nabla{(\phi\Sigma(\phi))}}\cdot\n=0$ on $S(t)$ at each time. With this condition, Eq.~\ref{eq:system-phi} is equivalent to Eq.~\ref{eq:mixture} on either sub-domain $\Omega_\alpha(t)$, being at the same time posed globally in $\Omega$. We will come back later (cf. Sect.~\ref{sect:interface-cond}) to the significance of such interface condition from the modeling viewpoint.

By comparing Eq.~\ref{eq:system-phi} with the standard mass balance equation of continuum mechanics we infer that the velocity $\v$ of the cellular matter is
\begin{equation}
	\v=-\kappa_m\nabla{(\phi\Sigma(\phi))}.
	\label{eq:vel.cell}
\end{equation}
Furthermore, by defining
\begin{equation}
	\Phi'(s):=s(s\Sigma(s))'
	\label{eq:Phiprime}
\end{equation}
we notice that Eq.~\ref{eq:system-phi} can be rewritten in the form of nonlinear diffusion:
\begin{equation}
	\frac{\partial\phi}{\partial t}-\kappa_m\Delta{\Phi(\phi)}=\Gamma(t,\,x,\,\phi,\,c).
	\label{eq:nonlin.diff}
\end{equation}

Conversely, Eq.~\ref{eq:mixture-c} is naturally defined on the whole $\Omega$: segregation is for cells, not for nutrient. Nevertheless, coherently with the segregation assumption, we redefine the right-hand side as
\begin{equation*}
	Q(t,\,x,\,\phi,\,c):=\sum_{\alpha=T,\,H}Q_\alpha(\phi,\,c)\1_{\Omega_\alpha(t)}(x)
\end{equation*}
and write
\begin{equation}
	\frac{\partial c}{\partial t}-D\Delta{c}=Q(t,\,x,\,\phi,\,c).
	\label{eq:system-c}
\end{equation}

\subsection{Aims and scope}
This paper is concerned with mathematical models of tumor growth of the kind outlined above, with a twofold goal. On the one hand, to discuss biologically consistent modeling lines for the phenomenological terms of the equations (namely, the functions $\Sigma$, $\Gamma$, and $Q$), as well as suitable boundary, interface, and initial conditions. On the other hand, to obtain qualitative results, such as \emph{a priori} nonnegativity, boundedness, and uniqueness of the solution, along with continuous dependence estimates, which support modeling with mathematical rigor. For this reason, the paper is ideally divided in two parts.

The first part, encompassing Sects.~\ref{sect:const.ass}, \ref{sect:conditions}, is especially devoted to modeling. Particularly, Sect.~\ref{sect:const.ass} surveys the most popular models proposed in the literature for $\Sigma$, $\Gamma$, and $Q$. Inspired by them, it fixes some modeling assumptions which will be used throughout the subsequent sections. Section~\ref{sect:conditions} discusses boundary, interface, and initial conditions needed to formulate mathematical problems, with special emphasis on the use of the former for simulating the surrounding environment e.g., a nearby vasculature.

The second part, encompassing Sects.~\ref{sect:notations}--\ref{sect:time.independent}, is targeted at analytical issues. Specifically, Sect.~\ref{sect:notations} is a preliminary technical one, introducing the main notations and recalling the essential theoretical background. Subsequently, Sects.~\ref{sect:time.dependent}, \ref{sect:time.independent} approach the time-dependent and time-independent problems, respectively, establishing \emph{a priori} estimates on their solutions. Existence of solutions is also explicitly addressed in the one-dimensional stationary case.

Finally, Sect.~\ref{sect:developments} sketches some research perspectives on recent multiphase models of tumor growth incorporating explicitly the attachment/detachment of cells to/from the extracellular matrix.

The paper is equipped with two Appendices, which contribute to make it as self-contained as possible. Appendix \ref{app:poisson} concerns the handling of the nonlinearities in the equations of the models. Appendix \ref{app:no.far} further extends the theory to other kinds of boundary conditions, partly different from those discussed in Sect.~\ref{sect:conditions}, relevant for applications.

\section{Constitutive assumptions}
\label{sect:const.ass}
\subsection{The cell stress function}
The function $\Sigma$ expresses the internal stress to each cell population. As already mentioned, the Cauchy excess stress tensors are given by $\T_T=\T_H=-\Sigma(\phi)\mathbf{I}$, hence $\Sigma$ acts as an intercellular pressure depending on the local cell packing.

For theoretic purposes, in the sequel it will be more customary to deal with the so-called \emph{constitutive function} $\Phi$, defined by Eq.~\ref{eq:Phiprime}, rather than with $\Sigma$ itself, although it is obviously possible to switch at any time to either function via the above-mentioned relationship.

Diffusion problems are well-known to be ill-posed if the diffusion coefficient is negative, therefore a very basic requirement in our case is that $\Phi'$ be nonnegative. A more complete characterization is provided by the following assumption:

\begin{enumerate}
\item[\bf (H1)] $\Phi:\R\to\R$ is smooth, strictly increasing, and normalized in such a way that $\Phi(0)=0$.
\end{enumerate}

Strict monotonicity of the constitutive function is classically required in the theory of nonlinear parabolic equations \cite{vazques2007pme-book}. Notice that if $\Phi$ is strictly increasing then $\Phi'$ cannot vanish but possibly at the origin (indeed $\Phi'(0)=0$ is forced by Eq.~\ref{eq:Phiprime} if ${(s\Sigma(s))}'$ is not infinite in $s=0$), therefore $\Phi'(s)>0$ for all $s\ne 0$ and $\Phi$ is invertible. We anticipate that we will use invertibility in Sect.~\ref{sect:time.independent} for the existence theory of the solutions to the stationary problem.

Many models in the literature assume that $\Sigma$ grows steeply as cells get highly packed. For example, in \cite{MR2257726} a particular instance of the following function is proposed (Fig. \ref{fig:Sigma1}, left):
\begin{equation}
	\Sigma(s)=as+b{\left[{(s-\phith)}^+\right]}^n,
	\label{eq:Sigma-tosambprez}
\end{equation}
where $a,\,b$ are positive constants with $b\gg a$, $n\geq 1$ is integer, $\phith\in(0,\,1)$ is the close-packing cell volume ratio, and $(\cdot)^+$ denotes the positive part of its argument. In practice, such a $\Sigma$ is a physiological pressure for normally packed cells, which rapidly increases as soon as the tissue becomes overly dense. Function \ref{eq:Sigma-tosambprez} is smooth for $n>1$ and piecewise smooth for $n=1$, and the corresponding constitutive function $\Phi$:
\begin{equation*}
	\Phi(s)=\frac{2}{3}as^3+b{\left[{(s-\phith)}^+\right]}^n\left\{s^2
		-\frac{s{(s-\phith)}^+}{n+1}+\frac{{[{(s-\phith)}^+]}^2}{(n+1)(n+2)}\right\}
\end{equation*}
fulfills hypothesis (H1).

Other authors use \cite{manoussaki2003mma,murray2003mtb} (Fig. \ref{fig:Sigma1}, center)
\begin{equation}
	\Sigma(s)=\frac{\tau s}{1+\lambda s^2},
	\label{eq:Sigma-manmurr}
\end{equation}
$\lambda,\,\tau>0$, which grows again linearly for small $s$ (physiological pressure) but then reproduces a release of the stress after the maximum $\tau/(2\sqrt{\lambda})$ attained for $s=1/\sqrt{\lambda}$. This should model a saturation effect due to that at high densities all cells are not able to push. It is interesting to note that such a behavior is opposite to the one assumed by Eq.~\ref{eq:Sigma-tosambprez}, nevertheless from Eq.~\ref{eq:Phiprime} it can be easily computed
\begin{equation*}
	\Phi(s)=\frac{\tau}{\lambda}\left(\frac{\arctan{(\sqrt{\lambda}s)}}{\sqrt{\lambda}}-
		\frac{s}{1+\lambda s^2}\right),
\end{equation*}
which complies in turn with hypothesis (H1).

\begin{figure*}
\centering
\includegraphics[width=\textwidth,clip]{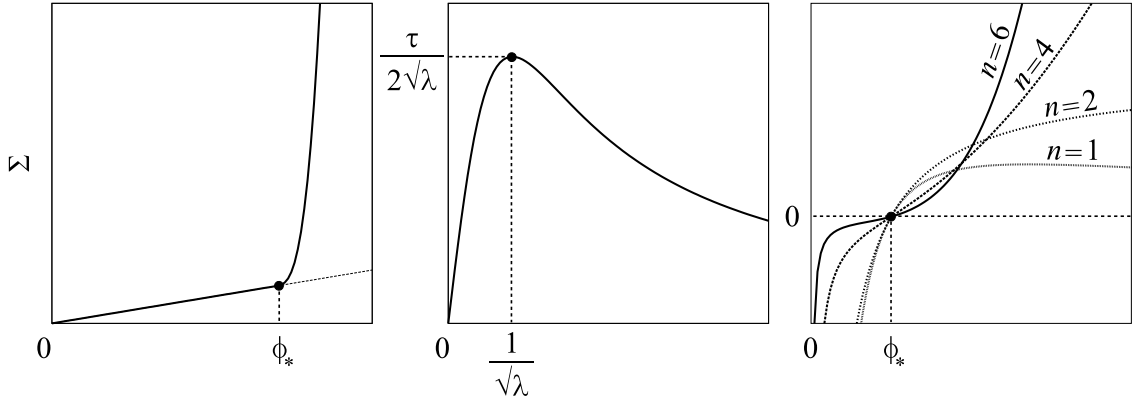} \quad
\caption{From left to right, the intercellular stress functions \ref{eq:Sigma-tosambprez}, \ref{eq:Sigma-manmurr}, \ref{eq:Sigma-tos}.}
\label{fig:Sigma1}
\end{figure*}

In the previous two examples it results $\Sigma(s)\geq 0$ for all $s\geq 0$, but this is not strictly necessary for hypothesis (H1) to be satisfied: even if $\Sigma(s)<0$ for some $s\geq 0$ one might get $\Phi'(s)>0$ for all $s\ne 0$. As recalled in \cite{MR1909425,byrne2003mst}, negative values of the stress model adhesive intercellular forces, which compete with the repulsive ones because cells, unless too packed, like to stick together to form multicellular aggregates. This behavior is reproduced, for instance, by the function $\Sigma$ proposed in \cite{MR2379886} (Fig. \ref{fig:Sigma1}, right):
\begin{equation}
	\Sigma(s)=
	\begin{cases}
		\dfrac{1}{s}\log{\left\vert\dfrac{s}{\phith}\right\vert} & \text{if\ } n=1 \\
		\\[-0.3cm]
		\dfrac{n}{n-1}\cdot\dfrac{{\vert s\vert}^{n-1}-\phith^{n-1}}{s} & \text{if\ } n>1,
	\end{cases}
	\label{eq:Sigma-tos}
\end{equation}
where now $\phith$ denotes the \emph{stress-free volume ratio} corresponding to unstressed tissue (i.e., $\Sigma(\phith)=0$). Notice that $\Sigma(s)\leq 0$ for $\vert s\vert\leq\phith$, with $\Sigma(s)\to -\infty$ for $\vert s\vert\to 0$. The resulting constitutive function:
\begin{equation*}
	\Phi(s)=\vert s\vert^{n-1}s
\end{equation*}
is strictly increasing for all $n\geq 1$ and turns Eq.~\ref{eq:nonlin.diff} into the porous medium equation with nonlinear forcing term. Function \ref{eq:Sigma-tos} somehow summarizes qualitatively the trends of the functions \ref{eq:Sigma-tosambprez}, \ref{eq:Sigma-manmurr} at large volume ratios, indeed for $1\leq n\leq 2$ it is bounded from above and, if $n<2$, tends to zero when $s\to +\infty$ (saturation effect), whereas for $n>2$ it grows unboundedly and more and more steeply as $n$ increases.

\begin{figure*}
\centering
\includegraphics[width=0.75\textwidth,clip]{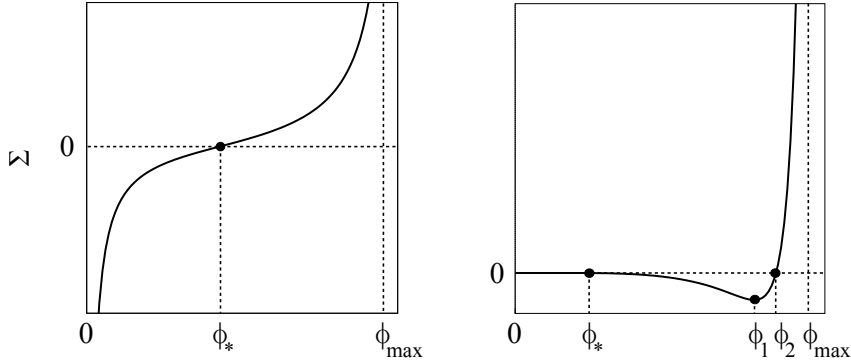}
\caption{From left to right, the intercellular stress function \ref{eq:Sigma-deanprez} and that used in \cite{byrne2003mst}.}
\label{fig:Sigma2}
\end{figure*}

In \cite{deangelis2000adm} the authors introduce the idea, next borrowed by a few other papers \cite{MR1909425,byrne2003mst,chaplain2006mml}, that the cell stress blows up when $\phi$ approaches a maximum allowed volume ratio $\phimax\in (0,\,1]$, which corresponds to an asymptote of the function $\Sigma$ for $s=\phimax$ (Fig \ref{fig:Sigma2}, left): 
\begin{equation}
	\Sigma(s)=p(\phimax-\phith)\frac{s-\phith}{\vert s\vert(\phimax-s)},
	\label{eq:Sigma-deanprez}
\end{equation}
where $p>0$ is a constant coefficient and $\phith$ denotes again the stress-free volume ratio. Notice that, $\Sigma$ being infinite at $\phimax$, $\Phi'$ is also infinite, which violates hypothesis (H1). However, strict monotonicity of $\Phi$ on $(-\infty,\,\phimax)$ is preserved by function \ref{eq:Sigma-deanprez}, while in general it fails on $[0,\,\phimax)$ with the function proposed in \cite{MR1909425,byrne2003mst} (Fig. \ref{fig:Sigma2}, right). The latter is such that $\Sigma(s)\leq 0$ for $\phith\leq s\leq\phi_2$ with a local minimum at $s=\phi_1\in(\phith,\,\phi_2)$ in order to take into account cell adhesiveness at low volume ratios, then $\Sigma(s)>0$ for $\phi_2<s<\phimax$ with $\Sigma(s)\to +\infty$ when $s\to\phimax$ to reproduce cell repulsion. In addition, they set $\Sigma(s)=0$ for $0\leq s\leq\phith$ to model that far apart cells ignore each other. We refer the reader to \cite{byrne2003mst} for the analytical expression of such a $\Sigma$. In \cite{chaplain2006mml} the adhesion region is instead eliminated by taking $\phith\equiv\phi_2$ and setting $\Sigma\equiv 0$ before $\phith$. In practice, the resulting function is the positive part of function \ref{eq:Sigma-deanprez}, which restores the monotonicity (however not strict) of $\Phi$ before $\phimax$.

It can be argued that condition $\Sigma(s)\to +\infty$ for $s\to\phimax$ is mainly intended to enforce the bound $\phi\leq\phimax$ on the solution to Eq.~\ref{eq:system-phi}, trusting to the physical intuition that a strong, infinite in the limit, intercellular pressure prevents cells from packing too much. Nevertheless, we will prove that this is not necessary to achieve the proper upper bound on $\phi$.

Before concluding the discussion on the constitutive function, we introduce a condition that will play a fundamental role in the forthcoming theory (cf. also \cite{fadimba1995ape,laurencot2005cmt}):
\begin{definition}
We say that a function $f:\R\to\R$ is \emph{$\Phi$-Lipschitz continuous} in an interval $I\subseteq\R$, with constant $\LipPhi{f}>0$, if
\begin{equation*}
	\vert f(s_2)-f(s_1)\vert^2\leq\LipPhi{f}\left(\Phi(s_2)-\Phi(s_1)\right)(s_2-s_1), \quad \forall\,s_1,\,s_2\in I.
\end{equation*}
\end{definition}
Notice that the right-hand side is nonnegative because $\Phi$ is increasing.

\subsection{The growth term}
The function $\Gamma$ expresses proliferation or death of cells in connection with the availability of nutrient. The basic principles inspiring the modeling of $\Gamma$ are usually that few cells proliferate less than many cells, that proliferation stops when cells fill all the available space, and that cells need a minimum amount of nutrient to survive.

For instance, in \cite{MR1920584} they use
\begin{equation*}
	\Gamma(\phi,\,c)=\phi(1-\phi)\frac{S_0 c}{1+S_1c}-\phi\frac{S_2+S_3c}{1+S_4c},
	\label{eq:g-breward}
\end{equation*}
where the first term at the right-hand side is the cell growth due to mitosis (fostered by the availability of nutrient), the second term is the cell death (enhanced by the lack of nutrient), and $S_0,\,\dots,\,S_4>0$ are parameters. Notice that cell growth is zero whenever $\phi=0$ (no cells) or $\phi=1$ (cells occupy the whole space, recall the saturation constraint) and that nutrient chemistry reminds of the Michaelis-Menten kinetic.

A simpler form of $\Gamma$, including an explicit nutrient threshold triggering the switch between cell proliferation and death, is that proposed in \cite{MR2379886}:
\begin{equation*}
	\Gamma(\phi,\,c)=\gamma\phi(1-\phi)(c-\cth)
	\label{eq:g-tosin}
\end{equation*}
where $\gamma,\,\cth>0$ are parameters. In this case, for $\phi\in[0,\,1]$ it results $\Gamma>0$ if $c>\cth$, $\Gamma<0$ if $c<\cth$, hence $\Gamma$ acts as a source or a sink, respectively, according to the values taken by $c$. Unfortunately, such a $\Gamma$ has the drawback of vanishing for $\phi=1$ regardless of $c$, which implies that cells no longer die once they have reached their maximum concentration, even if the nutrient falls below the survival threshold. This difficulty may be overcome by the following correction:
\begin{equation}
	\Gamma(\phi,\,c)=\gamma_1\phi(1-\phi)(c-\cth)^{+}-\gamma_2\phi(c-\cth)^{-},
	\label{eq:gamma-tosin-corr}
\end{equation}
where ${(\cdot)}^{+}$, ${(\cdot)}^{-}$ denote positive and negative parts of their arguments and $\gamma_1,\,\gamma_2>0$ are parameters.

In \cite{astanin2009mmw,MR2455396} proliferation and death of cells are linked to energy reasonings, specifically ATP consumption through nutrient oxidation along the cell cycle, and the following form of $\Gamma$ is proposed:
\begin{equation}
	\Gamma(\phi,\,c) = \frac{k\ln{2}}{Q^0_M}\phi(f(\phi)g(c)-\hat\theta)^{+}-
		\frac{k\ln{2}}{\hat\theta\tau_{1/2}}\phi(f(\phi)g(c)-\hat\theta)^{-},
	\label{eq:g-astanin}
\end{equation}
where $k,\,\hat\theta,\,Q^0_M,\,\tau_{1/2}>0$ are the reaction rate of oxidation, the total rate of ATP consumption, the average cost of the full cell cycle, and the half-life of dying cells, respectively. The function $f$ (of $\phi$ alone) is introduced to inhibit proliferation (namely, ATP production) in overly dense tissues, whereas the function $g$ (of $c$ alone) defines the effectiveness of the oxidative process in terms of nutrient supply. In more detail, $f$ vanishes when $\phi$ attains the maximum threshold $\phimax$ and $g$ increases with $c$. Examples of such functions are $f(\phi)=\phimax-\phi$, $g(c)=c$.

In \cite{chaplain2006mml,MR2471305} it is suggested that cell duplication and death may occur on a stress-induced basis, depending on the level of compression felt from the surrounding tissue:
\begin{equation}
	\Gamma_\alpha(\phi,\,c)=\gamma_\alpha\phi\Heav(\Sigma^\ast_\alpha-\Sigma(\phi))
		\left(\frac{c}{\cth}-1\right)-\delta_\alpha\phi\Heav(\Sigma(\phi)-\Sigma^{\ast\ast}_\alpha),
	\label{eq:gamma-stressind}
\end{equation}
$\Heav$ being (possibly a mollification of) the Heaviside function and $\Sigma_\alpha^\ast\leq\Sigma_\alpha^{\ast\ast}$ two stress thresholds. When the actual stress acting on the cells is above the first threshold proliferation is inhibited. If the stress further grows above the second threshold then cell apoptosis is triggered. Since the sensitivity to the stress affects the way in which a cell runs through its vital cycle, which is what mainly breaks down when mutations change a normal cell into an abnormal one, the thresholds $\Sigma^\ast_\alpha$, $\Sigma^{\ast\ast}_\alpha$ are expected to be different for tumor and host cells, in particular $\Sigma^{\ast(\ast)}_T>\Sigma^{\ast(\ast)}_H$ \cite{chaplain2006mml}.

Inspired by the examples above, we see that a convenient structure of the growth terms is
\begin{equation*}
	\Gamma_\alpha(\phi,\,c)=\sum_{\nu=p,\,d}\gamma_\alpha^\nu f_\alpha^\nu(\phi)g_\alpha^\nu(c)-\delta\phi,
\end{equation*}
which allows one to account for possible differences in the mechanisms of proliferation ($\nu=p$) and death ($\nu=d$) of tumor and host cells, as depicted for instance by Eqs.~\ref{eq:gamma-tosin-corr}, \ref{eq:gamma-stressind}. In more detail, the coefficients $\gamma_\alpha^\nu$ are specific proliferation and death rates for tumor and host cells. The terms $f_\alpha^\nu(\phi)g_\alpha^\nu(c)$ refer to joint stress-nutrient induced proliferation and apoptosis, under the assumption that the cell stress is directly determined by the volume ratio $\phi$. Finally, the term $-\delta\phi$ accounts for natural cell apoptosis without influence from the distribution of nutrient, the coefficient $\delta$ being the same for both cell populations.

Some technical requirements on the previous terms are now stated, taking into account a generic maximum volume ratio $\phimax\in (0,\,1]$ allowed for cell packing\footnote{If $\phimax<1$ then $1-\phimax$ is the constant volume ratio of the rigid non-remodeling extracellular matrix, and $\phi_\ell:=\phimax-\phi$ the volume ratio of the extracellular fluid (not explicitly modeled) filling the interstices within the mixture to enforce the saturation constraint.}:
\begin{enumerate}
\item[\bf (H2)] $\gamma_\alpha^p,\,\delta>0$, $\gamma_\alpha^d<0$
\item[\bf (H3)] $f_\alpha^\nu$ bounded, $\Phi$-Lipschitz continuous, and nonnegative in $[0,\,\phimax]$
\begin{enumerate}
\item[\bf (H3.1)] $f_\alpha^p\geq 0$ in $(-\infty,\,0)$, $f_\alpha^p\leq 0$ in $(\phimax,\,+\infty)$, $f_\alpha^p(0)=f_\alpha^p(\phimax)=0$
\item[\bf (H3.2)] $f_\alpha^d\leq 0$ in $(-\infty,\,0)$, $f_\alpha^d\geq 0$ in $(\phimax,\,+\infty)$, $f_\alpha^d(0)=0$
\item[\bf (H3.3)] $f_\alpha^d$ nondecreasing in $[0,\,\phimax]$
\end{enumerate}
\item[\bf (H4)] $g_\alpha^\nu$ locally bounded, Lipschitz continuous, and nonnegative in $\R$.
\end{enumerate}

Although not explicitly required, we incidentally notice that further assumptions may be suggested by physical considerations. For instance, nutrient-induced proliferation $g_\alpha^p$ may be non-decreasing and nutrient-induced death $g_\alpha^d$ non-increasing.

\subsection{The absorption term}
The function $Q$ describes the uptake of nutrient by cells. A common and simple prototype of this term is (see e.g., \cite{anderson1998cdm,byrne2003mst,MR2379886})
\begin{equation*}
	Q(\phi,\,c)=-\lambda\phi c,
\end{equation*}
where $\lambda>0$ is a parameter. This form translates the basic principle that the absorption of nutrient depends simultaneously on the number of cells present in the domain and on the quantity of nutrient available to them, including that few cells uptake, on the whole, few nutrient even if the latter is abundantly supplied, and, conversely, that few nutrient can poorly feed a large cell population. More general absorption terms are introduced in the series of papers \cite{MR2111919,MR2180715,MR2167659,bertuzzi2007cra}:
\begin{equation*}
	Q(\phi,\,c)=-\phi\varphi(c),
\end{equation*}
where $\varphi$ is then suggested to be of Michaelis-Menten type (cf. also \cite{MR1920584,manoussaki2003mma}). This allows the authors to generalize $Q$ to the case of multiple species of nutrients, including in $\varphi$ the dependence on the various concentrations \cite{MR2455391}.

In \cite{MR2455396} the absorption term is directly linked to the chemical mechanisms internal to the cells responsible for proliferation and death, and the following form of $Q$ is proposed:
\begin{equation*}
	Q(\phi,\,c)=-\lambda\phi f(\phi)g(c),
\end{equation*}
where $\lambda>0$ is the oxygen uptake rate and the functions $f,\,g$ are the same as in Eq.~\ref{eq:g-astanin}.

In view of these examples, and of possible differences in the consumption of nutrient by normal and abnormal cells, we refer to the following structure of the absorption term:
\begin{equation*}
	Q_\alpha(\phi,\,c)=-\lambda_\alpha h_\alpha(\phi)q_\alpha(c).
\end{equation*}
Here, $\lambda_\alpha$ is the specific absorption rate of the cell population $\alpha$ while the functions $h_\alpha$ account for cell-dependent uptake dynamics, which may differ for cancer and host cells because of different internal chemistry. Finally, $q_\alpha$ is the chemical consumption rate of nutrient, which instead is much likely to be the same for tumor and host cells as it depends essentially on the chemical properties of the environment and of the nutrient itself rather than on cell genetics.

Some technical assumptions on these terms are in order, namely:
\begin{enumerate}
\item[\bf (H5)] $\lambda_\alpha>0$
\item[\bf (H6)] $h_\alpha$ $\Phi$-Lipschitz continuous and nonnegative in $\R$
\item[\bf (H7)] $q_\alpha$, locally bounded and nonnegative in $[0,\,+\infty)$
\begin{enumerate}
\item[\bf (H7.1)] $q_\alpha(0)=0$, $q_\alpha\leq 0$ in $(-\infty,\,0)$
\item[\bf (H7.2)] $q_\alpha$ nondecreasing in $[0,\,+\infty)$.
\end{enumerate}
\end{enumerate}

Again, additional assumptions, not strictly needed for theoretic issues, may be welcome for physical consistency. For instance, one may require the cell-dependent absorption rate to vanish when no cells are present, i.e., $h_\alpha(0)=0$. We anticipate that we will actually use this assumption when addressing the existence of stationary solutions.

\section{Boundary, interface, and initial conditions}
\label{sect:conditions}
\subsection{Boundary conditions}
\label{sect:bc}
The parabolic character of Eqs.~\ref{eq:nonlin.diff}, \ref{eq:system-c} calls for conditions on the whole boundary $\partial\Omega\times(0,\,\Tmax]$ in order to properly define the corresponding mathematical problems. The most common boundary conditions in tumor growth problems refer to characteristic values of some quantities at the periphery of the cell tissue (Dirichlet conditions) or to their fluxes across the external shell of the portion of tissue under consideration (Neumann or Robin conditions). Often boundary conditions account for interactions of the cell aggregate with the outer environment (not explicitly modeled).

Let us consider, at first, the chemicals nourishing the cells. Nutrient is supplied to the cells by the external environment e.g., by a nearby vasculature, whence, dissolved in the extracellular fluid, it diffuses through the cell tissue. The presence of the vasculature can be modeled as a boundary condition for Eq.~\ref{eq:system-c}: one identifies a portion of the boundary, say $\partial_b\Omega$ the subscript $b$ standing for ``blood'', with a blood vessel, whence the nutrient carried by blood flows into the tissue. One has therefore a condition on the normal flux $-D\nabla{c}\cdot\n$, $\n$ being the outward normal unit vector to $\partial_b\Omega$. This may be either a Neumann condition, if the flux is given, or a Robin condition, if the flux is expressed in terms of $c$ itself, for instance as $\eta(c-c_b)$ where $\eta,\,c_b>0$ are parameters. Such a condition implies that the flow of nutrient across the vessel wall depends on the quantity of nutrient already present in the tissue, with respect to a characteristic threshold $c_b$. In more detail, if the concentration $c$ equals the characteristic concentration $c_b$ then there is no flux, while if $c<c_b$ then the nutrient flows from the blood to the tissue. Conversely, if $c>c_b$ then the nutrient flows from the tissue to the blood, namely blood carries away the nutrient in excess. However, this situation is not expected to happen because blood is the only source of nutrient in this model. The parameter $c_b$ can be identified with the average physiological concentration of nutrient in the blood, whereas $\eta$ is a characteristic speed of filtration through the vessel wall. We point out that Robin's is the biologically most appropriate condition to be imposed at the vessel wall, see \cite{keener1998mp,truskey2009tpb}, and in this respect Neumann condition should be regarded as a zeroth-order approximation. For this reason, we ultimately set
\begin{equation*}
	-D\nabla{c}\cdot\n=\eta(c-c_b) \quad \text{on\ } \partial_b\Omega\times(0,\,\Tmax].
\end{equation*}

In the remaining portion of the boundary, say $\partial_f\Omega$ the subscript $f$ standing for ``far'' (in the sense that this boundary is far from the vasculature where tumor growth mainly takes place), one might prescribe the concentration $c_b$ for conveying the idea that the quantity of nutrient is there the average physiological one in healthy tissues:
\begin{equation*}
	c=c_b \quad \text{on\ } \partial_f\Omega\times(0,\,\Tmax],
\end{equation*}
which is a Dirichlet boundary condition.

We suppose $\partial_b\Omega\cup\partial_f\Omega=\partial\Omega$, $\operatorname{Int}{\partial_b\Omega}\cap\operatorname{Int}{\partial_f\Omega}=\emptyset$ for consistency (cf. Fig. \ref{fig:domain}) but we do not exclude that either $\partial_f\Omega$ or $\partial_b\Omega$ is empty, if for instance the vasculature fully surrounds the cellular tissue ($\partial_f\Omega=\emptyset$, $\partial\Omega\equiv\partial_b\Omega$) or if one is concerned with avascular tumors ($\partial_b\Omega=\emptyset$, $\partial\Omega\equiv\partial_f\Omega$).

\begin{figure*}
\centering
\includegraphics[width=\textwidth,clip]{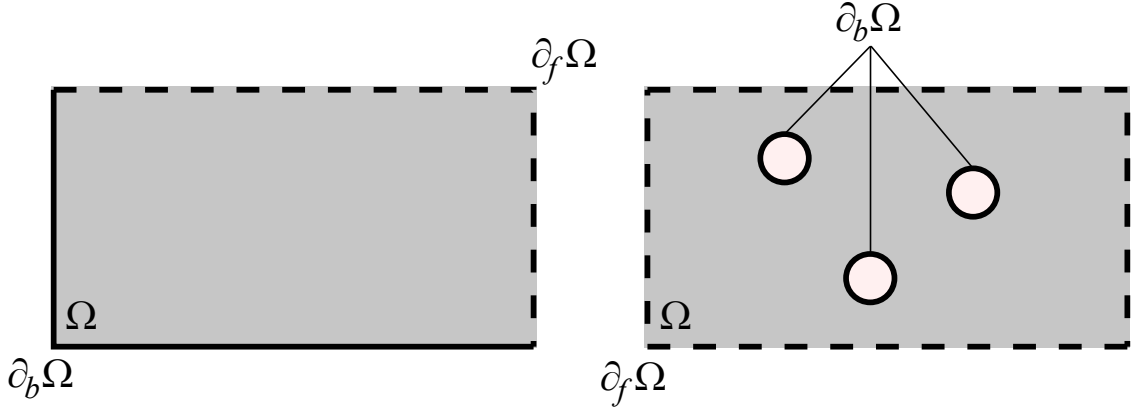} \qquad
\caption{Left: blood vessels and far boundaries are among the edges of $\Omega$. Right: blood vessels are internal holes to $\Omega$, whereas the far boundary coincides with the outer edges of $\Omega$.}
\label{fig:domain-edges}
\end{figure*}

Concerning the cells, since we are considering \emph{in situ} tumor growth, it is reasonable to assume that they do not penetrate the vasculature, which corresponds to no flux across the boundary $\partial_b\Omega$:
\begin{equation*}
	\kappa_m\nabla{\Phi(\phi)}\cdot\n=0 \quad \text{on\ } \partial_b\Omega\times(0,\,\Tmax].
\end{equation*}
Recalling Eq.~\ref{eq:vel.cell}, this condition can be rewritten as $\phi\v\cdot\n=0$, which says that the normal component of the velocity of the cells vanishes at the blood vessel. Cells neither cross the vessel wall nor detach from it, but they can slide along the vessel because no restriction is imposed on the tangential component of their velocity. Instead, at the far boundary $\partial_f\Omega$ the tissue is essentially relaxed, for the main dynamics is expected to occur near the vasculature. Therefore, considering that in the present setting the stress state is determined by the cell volume ratio, the natural condition for $\phi$ on $\partial_f\Omega$ is some physiological value
\begin{equation*}
	\phi=\phith\in(0,\,\phimax) \quad \text{on\ } \partial_f\Omega\times(0,\,\Tmax],
\end{equation*}
possibly coinciding with the stress-free value if the function $\Sigma$ admits one.

Boundary conditions outlined here are mainly indicative, and may be detailed more precisely for specific domains $\Omega$. In particular, the boundaries $\partial_b\Omega$, $\partial_f\Omega$ need not be connected, especially when several blood vessels are present (in which case $\partial_b\Omega$ will be presumably the union of several connected components). As an example, one may consider the applications illustrated in \cite{MR2471305}, which deal with two-dimensional tumor growth in a rectangular domain $\Omega$. In some cases, blood vessels coincide with one or more (not necessarily adjacent) edges of $\Omega$, the remaining ones forming instead $\partial_f\Omega$ (Fig. \ref{fig:domain-edges}, left). In other cases, all of the outer edges of $\Omega$ define $\partial_f\Omega$, whereas blood vessels are circular holes within the rectangle. The union of their circumferences is then $\partial_b\Omega$, which plays the role of an internal boundary to $\Omega$ (Fig. \ref{fig:domain-edges}, right).

\subsection{Interface conditions}
\label{sect:interface-cond}
The interface $S(t)$ separating the two sub-domains $\Omega_T(t)$, $\Omega_H(t)$ is a material one for the cells, meaning that the latter cannot detach from it on either side. This entails the continuity of their normal velocity across $S(t)$, i.e., recalling Eq.~\ref{eq:vel.cell},
\begin{equation}
	\jump{\v\cdot\n}=-\kappa_m\jump{\nabla{(\phi\Sigma(\phi))}}\cdot\n=0, \quad \forall\,t\in(0,\,\Tmax].
	\label{eq:interface-contv}
\end{equation}
In addition, classical theory of continuum mechanics requires the continuity of the stress of the interfacing materials:
\begin{equation*}
	\jump{\phi\Sigma(\phi)\n}=0, \quad \forall\,t\in(0,\,\Tmax].
	\label{eq:interface-cont}
\end{equation*}
For a continuous stress function $\Sigma$, this is fulfilled if $\jump{\phi}=0$, which, coupled with Eq.~\ref{eq:interface-contv}, yields $\kappa_m\jump{\phi\nabla{(\phi\Sigma(\phi))}}\cdot\n=0$. Hence the condition necessary for the validity of Eq.~\ref{eq:system-phi} on $\Omega$ is compatible with the standard interface conditions of continuum mechanics. Owing to this argument, we ultimately impose
\begin{equation*}
	\kappa_m\jump{\phi\nabla{(\phi\Sigma(\phi))}}\cdot\n=0 \quad \text{on\ } S(t),\,t\in(0,\,\Tmax],
\end{equation*}
because it is this interface condition which is really needed in our problem.

In the sequel we will assume that the time evolution of the interface $S(t)$, and consequently those of the sub-domains $\Omega_T(t)$, $\Omega_H(t)$, is given. In other words, in the subsequent \emph{a priori} estimates we will disregard the explicit coupling of their dynamics with the ones of the cells. We defer to a forthcoming work a more detailed investigation of such additional issues of the problem.

\subsection{Initial conditions}
The initial distributions of cells and nutrient in $\Omega$ are described by two functions $\phi_0,\,c_0:\Omega\to\R$, which should not exceed the expected physiological ranges. Therefore:
\begin{enumerate}
\item[\bf (H8)] $0\leq\phi_0\leq\phimax$, $0\leq c_0\leq c_b$ in $\Omega$.
\end{enumerate}
In particular, as observed in Sect.~\ref{sect:bc}, nutrient concentration is expected to stay below the average physiological value $c_b$ because, in the present context, blood is the only source of nutrient, which is then consumed by cells.

\section{Notations and theoretic background}
\label{sect:notations}
In this section we prepare to tackle the analysis of the above-discussed models. We fix the main notations and quickly recall some essential technical material.

\paragraph{Domain.} The spatial domain is an open and bounded set $\Omega\subset\R^d$ (physically $d=1,\,2,\,3$) with smooth boundary $\partial\Omega$. We use $\sigma$ for the coordinate along $\partial\Omega$ and $d\sigma$ for the $(d-1)$-dimensional Hausdorff measure in $\R^d$. The time interval of interest is $[0,\,\Tmax]$, with finite final time $\Tmax>0$. We denote by $\QT$ the cylinder $\Omega\times(0,\,\Tmax]\subset\R^{d+1}$.

\paragraph{Functions.} In general, we regard $\phi,\,c$ as functions parameterized by time. Particularly, when we want to emphasize the dependence on $t$ we write $\phi(t),\,c(t)$ for the functions of $x$ defined as $(\phi(t))(x)=\phi(t,\,x)$, $(c(t))(x)=c(t,\,x)$.

We will occasionally denote the time derivative by the subscript $t$ for short (e.g., $\phi_t=\partial_t\phi$, $c_t=\partial_t c$).

We write $\Lip{u}$ for the Lipschitz constant of a function $u$, and $u^+$, $u^-$ for its positive and negative part:
$u^+:=\max\{u,\,0\}$, $u^-:=\max\{0,\,-u\}$.

The indicator function of a set $A$ is $\1_A$, i.e., $\1_A(x)=1$ if $x\in A$, $\1_A(x)=0$ if $x\not\in A$.

\paragraph{Function spaces.} $L^2(\Omega)$ is the usual Hilbert space of square-integrable functions in $\Omega$, endowed with the norm
\begin{equation*}
	\norm{\cdot}_0:=\left(\intl_\Omega\abs{\cdot(x)}^2\,dx\right)^{1/2}.
\end{equation*}
We will also use a weaker $L^2$ norm denoted by $\wnorm{\cdot}_0$ (cf. Appendix \ref{app:poisson}).

$H^1(\Omega)$ is the Sobolev space of $L^2$ functions with square-integrable (weak) derivatives in $\Omega$, endowed with the norm
\begin{equation*}
	\norm{\cdot}_1:=\left(\norm{\cdot}_0^2+\norm{\nabla{\cdot}}_0^2\right)^{1/2}.
\end{equation*}
For $u\in H^1(\Omega)$, Stampacchia's Theorem asserts that $u^+,\,u^-\in H^1(\Omega)$ as well, with $\nabla{u^+}=\nabla{u}\1_{\{u>0\}}$, $\nabla{u^-}=-\nabla{u}\1_{\{u<0\}}$.

$H^1_{0,\B}(\Omega)$ is the space of $H^1$ functions whose trace vanishes on $\B\subseteq\partial\Omega$. The $L^2$ norm of the trace along $\B$ is
\begin{equation*}
	\norm{\cdot}_{0,\B}:=\left(\intl_\B\abs{\cdot}^2\,d\sigma\right)^{1/2}.
\end{equation*}
We will deal, in particular, with the case $\B=\partial_b\Omega$.

$L^2(0,\,\Tmax;\,H^1(\Omega))$ is the space of functions of $t$, valued in $H^1(\Omega)$, which are square-integrable in the interval $[0,\,\Tmax]$, endowed with the norm
\begin{equation*}
	\norm{\cdot}_{L^2_tH^1_x}:=\left(\intl_0^{\Tmax}\norm{\cdot(t)}_1^2\,dt\right)^{1/2}.
\end{equation*}
The spaces $L^2(0,\,\Tmax;\,L^2(\Omega))$, $L^2(0,\,\Tmax;\,H^1_{0,\partial_f\Omega}(\Omega))$ are defined analogously, and their respective norms denoted similarly. In particular, in the former we will use the norm
\begin{equation*}
	\wnorm{\cdot}_{L^2_tL^2_x}:=\left(\intl_0^{\Tmax}\wnorm{\cdot(t)}_0^2\,dt\right)^{1/2},
\end{equation*}
$\wnorm{\cdot}_0$ being defined in Appendix \ref{app:poisson}.

We introduce the following shorthand notations:
\begin{itemize}
\item $\VT:=L^2(0,\,\Tmax;\,H^1(\Omega))\times L^2(0,\,\Tmax;\,H^1(\Omega))$
\item $\V:=H^1(\Omega)\times H^1(\Omega)$
\item $V_f:=H^1_{0,\partial_f\Omega}(\Omega)$
\item $V_f'$ for the dual space of $V_f$.
\end{itemize}

We use the symbol $\dual{\cdot}{\cdot}$ for the duality pairing between $V_f$ and $V_f'$. Given $u\in L^2(0,\,\Tmax;\,V_f)$ with $u_t\in L^2(0,\,\Tmax;\,V_f')$, it results
\begin{equation*}
	\dual{u_t(t)}{u(t)}=\frac{1}{2}\frac{d}{dt}\norm{u(t)}_0^2.
\end{equation*}

We use the abbreviation ``a.e.'' for properties which hold ``almost everywhere'' with respect to the Lebesgue measure.

If $I\subseteq\R$ is an interval, $C^0(I)$ is the space of continuous functions in $I$, endowed with the norm
\begin{equation*}
	\norm{\cdot}_\infty:=\sup_{x\in I}\abs{\cdot(x)}.
\end{equation*}

\paragraph{Estimates and constants.} We write
\begin{equation*}
	a\lesssim b \quad \text{to mean} \quad \exists\,C>0\,:\,a\leq Cb,
\end{equation*}
the constant $C$ being understood to be independent of both $a$ and $b$, when the specific value of $C$ is unimportant. In this case, $C$ may vary from line to line within the same computation without explicit notice.

\paragraph{Inequalities.}
\begin{itemize}
\item Cauchy's: for all $a,\,b\in\R$, $ab\leq\frac{\epsilon a^2}{2}+\frac{b^2}{2\epsilon}$ with arbitrary $\epsilon>0$.
\item Poincar\'e's: if $u\in H^1_{0,\B}(\Omega)$ then $\norm{u}_0\lesssim\norm{\nabla{u}}_0$. The constant involved in this estimate, denoted by $C_P$, depends in general on $\Omega$ and $\B$.
\item Cauchy-Schwartz's: $\abs{\int_\Omega u(x)v(x)\,dx}\leq\norm{u}_0\norm{v}_0$ for all $u,\,v\in L^2(\Omega)$. 
\end{itemize}

\section{The time-dependent problem}
\label{sect:time.dependent}
In this section we consider the initial/boundary-value problem
\begin{equation}
	\left\{
	\begin{array}{rcll}
		\dfrac{\partial\phi}{\partial t}-\kappa_m\Delta{\Phi(\phi)} & = & \Gamma(t,\,x,\,\phi,\,c) & \text{in\ }
			\begin{cases}
				\Omega_T(t),\,\Omega_H(t) \\
				t\in (0,\,\Tmax]
			\end{cases}
			\\
		\dfrac{\partial c}{\partial t}-D\Delta{c} & = & Q(t,\,x,\,\phi,\,c) & \text{in\ } \QT \\[0.5cm]
		\begin{array}{rcl}
			\kappa_m\nabla\Phi(\phi)\cdot\n & = & 0, \\
			\phi & = & \phith, \\
			\kappa_m\jump{\nabla{\Phi(\phi)}}\cdot\n & = & 0 \\
			\phi(0) & = & \phi_0,
		\end{array} & &
		\begin{array}{rcl}
			-D\nabla{c}\cdot\n & = & \eta(c-c_b) \\
			c & = & c_b \\
			\\
			c(0) & = & c_0
		\end{array} & 
		\begin{array}{l}
			\text{on\ } \partial_b\Omega\times(0,\,\Tmax] \\
 			\text{on\ } \partial_f\Omega\times(0,\,\Tmax] \\
			\text{on\ } S(t),\,t\in(0,\,\Tmax]\\
			\text{in\ } \Omega
		\end{array}
	\end{array}
	\right.
	\label{eq:parab}
\end{equation}
along with the series of hypotheses (H1)--(H8), and we look for estimates of nonnegativity, boundedness, uniqueness, and continuous dependence on the data of the functions $\phi,\,c$. We assume that solutions exist to this problem, in the suitable sense specified below. Notice that, in \ref{eq:parab}, the interface condition has been conveniently rewritten in terms of $\nabla{\Phi(\phi)}$ using Eq.~\ref{eq:Phiprime}.

\begin{definition}[Weak solutions for the time-dependent problem]
A weak solution to Problem~\ref{eq:parab} is a pair $(\phi,\,c)\in\VT$ such that:
\begin{enumerate}
\item[(i)] $\phi_t,\,c_t\in L^2(0,\,\Tmax;\,V_f')$
\item[(ii)] $\Phi(\phi)\in L^2(0,\,\Tmax;\,H^1(\Omega))$
\item[(iii)] $\phi=\phith$, $c=c_b$ on $\partial_f\Omega\times(0,\,\Tmax]$ in the trace sense
\item[(iv)] $\phi(0)=\phi_0\in L^2(\Omega)$, $c(0)=c_0\in L^2(\Omega)$
\end{enumerate}
which satisfies
\begin{align}
	& \dual{\phi_t}{v_1}+\dual{c_t}{v_2} \nonumber \\
	& \phantom{\dual{\phi_t}{v_1}} +\intl_\Omega\left(\kappa_m\nabla{\Phi(\phi)}\cdot\nabla{v_1}+
		D\nabla{c}\cdot\nabla{v_2}\right)\,dx+\eta\intl_{\partial_b\Omega}(c-c_b)v_2\,d\sigma \nonumber \\
	& =\sum_{\alpha=T,\,H}\intl_{\Omega_\alpha(t)}\left(\sum_{\nu=p,\,d}\gamma_\alpha^\nu f_\alpha^\nu(\phi)
		g_\alpha^\nu(c)v_1-\delta\phi v_1-\lambda_\alpha h_\alpha(\phi)q_\alpha(c)v_2\right)\,dx
	\label{eq:parab-weak}
\end{align}
for all $v_1,\,v_2\in V_f$ and a.e. $t\in[0,\,\Tmax]$.
\end{definition}

\subsection{Nonnegativity and boundedness of the solution}
\begin{theorem}
Any weak solution $(\phi,\,c)\in\VT$ to Problem~\ref{eq:parab} satisfies
\begin{equation*}
	0\leq\phi(t,\,x)\leq\phimax, \quad 0\leq c(t,\,x)\leq c_b \quad \text{for a.e.\ } (x,\,t)\in\QT.
\end{equation*}
\label{theo:boundedness}
\end{theorem}
\begin{proof}
First we establish that $\phi,\,c$ are a.e. nonnegative by showing $\phi^-,\,c^-=0$ a.e. in $\QT$. Choosing $v_1=\phi^-(t)$, $v_2=c^-(t)$ as test functions in Eq.~\ref{eq:parab-weak} reveals
\begin{align*}
	& -\frac{1}{2}\frac{d}{dt}\left(\norm{\phi^-(t)}_0^2+\norm{c^-(t)}_0^2\right)
		-\kappa_m\intl_\Omega\Phi'(-\phi^-)\abs{\nabla{\phi^-}}^2\,dx-D\norm{\nabla{c}^-(t)}_0^2 \\
	& \qquad -\eta\norm{c^-(t)}_{0,\partial_b\Omega}^2-\eta c_b\intl_{\partial_b\Omega}c^-\,d\sigma \\
	& =\!\!\!\sum_{\alpha=T,\,H}\intl_{\Omega_\alpha(t)}\!\!\!\left(\sum_{\nu=p,\,d}\gamma_\alpha^\nu f_\alpha^\nu(-\phi^-)
		g_\alpha^\nu(c)\phi^{-}+\delta(\phi^-)^2-\lambda_\alpha h_\alpha(\phi)q_\alpha(-c^-)c^-\right)\! dx.
\end{align*}

Because of hypotheses (H2)--(H7), the right-hand side is nonnegative for a.e. $t\in[0,\,\Tmax]$. Integrating from $0$ to $t\leq\Tmax$ and using $\phi_0^-,\,c_0^-=0$ (hypothesis (H8)) we get then
\begin{multline*}
	\norm{\phi^-(t)}_0^2+\norm{c^-(t)}_0^2+2\kappa_m\intl_0^t\intl_\Omega\Phi'(-\phi^-)\abs{\nabla{\phi^-}}^2\,dx\,d\tau \\
	+2D\intl_0^t\norm{\nabla{c^-}(\tau)}_0^2\,d\tau
		+2\eta\intl_0^t\norm{c^-(\tau)}_{0,\partial_b\Omega}^2\,d\tau+2\eta c_b\intl_{\partial_b\Omega}c^-\,d\sigma\leq 0
\end{multline*}
for all $t\in[0,\,\Tmax]$, whence $\phi^-,\,c^-=0$ a.e. in $\QT$ due to the nonnegativity of each term at the left-hand side (use hypothesis (H1) for the term containing $\Phi'$).

\bigskip

Next we claim $(\phi-\phimax)^+=(c-c_b)^+=0$ a.e. in $\QT$, which amounts to $\phi\leq\phimax$, $c\leq c_b$. Let $\tilde{\phi}:=(\phi-\phimax)^+$, $\tilde{c}:=(c-c_b)^+$ for brevity. Taking $v_1=\tilde{\phi}(t)$, $v_2=\tilde{c}(t)$ as test functions in Eq.~\ref{eq:parab-weak} yields
\begin{align*}
	& \frac{1}{2}\frac{d}{dt}\left(\norm{\tilde{\phi}(t)}_0^2+\norm{\tilde{c}(t)}_0^2\right)
		+\kappa_m\intl_\Omega\Phi'(\phimax+\tilde{\phi})\abs{\nabla{\tilde{\phi}}}^2\,dx+D\norm{\nabla{\tilde{c}}(t)}^2_0 \\
	& \qquad +\eta\norm{\tilde{c}(t)}_{0,\partial_b\Omega}^2 \\
	& =\sum_{\alpha=T,\,H}\intl_{\Omega_\alpha(t)}\Biggl(\sum_{\nu=p,\,d}\gamma_\alpha^\nu
		f_\alpha^\nu(\phimax+\tilde{\phi})g_\alpha^\nu(c)\tilde{\phi}-\delta(\phimax+\tilde{\phi})\tilde{\phi} \\
	& \qquad -\lambda_\alpha h_\alpha(\phi)q_\alpha(c_b+\tilde{c})\tilde{c}\Biggr)\,dx,
\end{align*}
the right-hand side being this time nonpositive for a.e. $t\in[0,\,\Tmax]$. Integrating in time and using now $\tilde{\phi}(0)=\tilde{c}(0)=0$ we obtain
\begin{multline*}
	\norm{\tilde{\phi}(t)}_0^2+\norm{\tilde{c}(t)}_0^2+
		2\kappa_m\intl_0^t\intl_\Omega\Phi'(\phith+\tilde{\phi})\abs{\nabla{\tilde{\phi}}}^2\,dx\,d\tau \\
	+2D\intl_0^t\norm{\nabla{\tilde{c}}(\tau)}_0^2\,d\tau
		+2\eta\intl_0^t\norm{\tilde{c}(\tau)}_{0,\partial_b\Omega}^2\,d\tau\leq 0
\end{multline*}
for all $t\in[0,\,\Tmax]$, whence the claim follows by arguing like in the previous point.
\end{proof}

\subsection{Uniqueness and continuous dependence on the initial data}
\begin{theorem}
Let $(\phi_i,\,c_i)\in\VT$, $i=1,\,2$, be two weak solutions of Problem~\ref{eq:parab} corresponding to the initial conditions $(\phi_{i,0},\,c_{i,0})\in L^2(\Omega)\times L^2(\Omega)$. Then
\begin{align}
	\wnorm{\phi_2-\phi_1}_{L^2_tL^2_x}^2 &+
		\intl_0^{\Tmax}\intl_\Omega\left(\Phi(\phi_2)-\Phi(\phi_1)\right)(\phi_2-\phi_1)\,dx\,dt \nonumber \\
	& +\norm{c_2-c_1}_{L^2_tH^1_x}^2+\intl_0^{\Tmax}\norm{(c_2-c_1)(t)}_{0,\partial_b\Omega}^2\,dt \nonumber \\
	& \lesssim\norm{\phi_{2,0}-\phi_{1,0}}_0^2+\norm{c_{2,0}-c_{1,0}}_0^2.
	\label{eq:apriori}
\end{align}
In particular, the solution corresponding to a given initial condition is unique.
\label{theo:uniqueness}
\end{theorem}
\begin{proof}
It is sufficient to prove the estimate \ref{eq:apriori}, for then uniqueness easily follows out of it with $\phi_{1,0}=\phi_{2,0}$ and $c_{1,0}=c_{2,0}$.

For all $v_1,\,v_2\in V_f$, the difference $(\phi_2-\phi_1,\,c_2-c_1)$ of the two given solutions satisfies
\begin{align}
	& \dual{(\phi_2-\phi_1)_t}{v_1}+\dual{(c_2-c_1)_t}{v_2} \nonumber \\
	& \phantom{\dual{(\phi_2-\phi_1)_t}{v_1}}
		+\intl_\Omega\left(\kappa_m\nabla{(\Phi(\phi_2)-\Phi(\phi_1))}\cdot\nabla{v_1}+D\nabla{(c_2-c_1)}\cdot\nabla{v_2}\right)\,dx \nonumber \\
	& \phantom{\dual{(\phi_2-\phi_1)_t}{v_1}}
		+\eta\intl_{\partial_b\Omega}(c_2-c_1)v_2\,d\sigma \notag \\
	& =\!\!\!\sum_{\alpha=T,\,H}\intl_{\Omega_\alpha(t)}\!\!\Biggl(
		\sum_{\nu=p,\,d}\gamma_\alpha^\nu\{[f_\alpha^\nu(\phi_2)-f_\alpha^\nu(\phi_1)]g_\alpha^\nu(c_2)
			+f_\alpha^\nu(\phi_1)[g_\alpha^\nu(c_2)-g_\alpha^\nu(c_1)]\}v_1 \nonumber \\
	& \phantom{=\sum_{\alpha=T,\,H}\intl_{\Omega_\alpha(t)}\Biggl(}
		-\delta(\phi_2-\phi_1)v_1-\lambda_\alpha(h_\alpha(\phi_2)-h_\alpha(\phi_1))q_\alpha(c_2)v_2 \nonumber \\
	& \phantom{=\sum_{\alpha=T,\,H}\intl_{\Omega_\alpha(t)}\Biggl(}
		+\lambda_\alpha h_\alpha(\phi_1)(q_\alpha(c_2)-q_\alpha(c_1))v_2\Biggr)\,dx,
	\label{eq:proof-twosol}
\end{align}
whence, choosing the test functions $v_1=\P(\phi_2-\phi_1)$, $v_2=c_2-c_1$ and using Eq.~\ref{eq:action_of_P} (cf. Appendix~\ref{app:poisson}), we rewrite the left-hand side as
\begin{align}
	\text{l.h.s of (\ref{eq:proof-twosol})} &= \frac{1}{2}\frac{d}{dt}\left(\wnorm{(\phi_2-\phi_1)(t)}_0^2+
		\norm{(c_2-c_1)(t)}_0^2\right) \nonumber \\
	& \phantom{=}+\kappa_m\intl_\Omega(\Phi(\phi_2)-\Phi(\phi_1))(\phi_2-\phi_1)\,dx  \nonumber \\
	& \phantom{=}+D\norm{\nabla{(c_2-c_1)}(t)}_0^2+\eta\norm{(c_2-c_1)(t)}_{0,\partial_b\Omega}^2.
	\label{eq:proof-lhs}
\end{align}

As for the right-hand side, Cauchy-Schwartz's inequality for the standard inner product in $L^2(\Omega)$, along with the boundedness of $f_\alpha^\nu$, \dots, $q_\alpha$ in the ranges of $\phi,\,c$ (recall Theorem \ref{theo:boundedness}), allows us to bound it from above as
\begin{align*}
	\text{r.h.s of (\ref{eq:proof-twosol})} &\leq\sum_{\substack{\alpha=T,\,H \\ \nu=p,\,d}}
		\gamma_\alpha^\nu(\norm{g_\alpha^\nu}_\infty\norm{f_\alpha^\nu(\phi_2)-f_\alpha^\nu(\phi_1)}_0 \\
	& \phantom{\leq}+\norm{f_\alpha^\nu}_\infty\norm{g_\alpha^\nu(c_2)-g_\alpha^\nu(c_1)}_0)\norm{\P(\phi_2-\phi_1)(t)}_0 \\
	& \phantom{\leq} -\delta\wnorm{(\phi_2-\phi_1)(t)}_0^2 \\
	& \phantom{\leq} +\sum_{\alpha=T,\,H}\lambda_\alpha\norm{q_\alpha}_\infty
		\norm{h_\alpha(\phi_2)-h_\alpha(\phi_1)}_0\norm{(c_2-c_1)(t)}_0,
\end{align*}
where we have further used that $-\lambda_\alpha h_\alpha(\phi_1)(q_\alpha(c_2)-q_\alpha(c_1))(c_2-c_1)\leq 0$ a.e. in $\QT$ because $q_\alpha$ is nondecreasing (hypothesis (H7.2)). Since $f_\alpha^\nu$, $h_\alpha$ are $\Phi$-Lipschitz continuous, it results
\begin{equation*}
	\norm{f_\alpha^\nu(\phi_2)-f_\alpha^\nu(\phi_1)}_0^2,\,\norm{h_\alpha(\phi_2)-h_\alpha(\phi_1)}_0^2
		\lesssim\intl_\Omega\left(\Phi(\phi_2)-\Phi(\phi_1)\right)(\phi_2-\phi_1)\,dx,
\end{equation*}
Combing this with the Lipschitz continuity of the $g_\alpha^\nu$'s (hypothesis (H4)), Cauchy's inequality, and the fact that $\norm{\P\cdot}_0\lesssim\wnorm{\cdot}_0$ (cf. Appendix \ref{app:poisson}), after some algebraic manipulations we arrive at
\begin{align}
	\text{r.h.s. of (\ref{eq:proof-twosol})} &\lesssim
		\epsilon\intl_\Omega\left(\Phi(\phi_2)-\Phi(\phi_1)\right)(\phi_2-\phi_1)\,dx \nonumber \\
	& \phantom{\leq} +\left(\frac{1}{2}+\frac{1}{\epsilon}\right)\left(\wnorm{(\phi_2-\phi_1)(t)}_0^2+\norm{(c_2-c_1)(t)}_0^2\right)
	\label{eq:proof-rhs}
\end{align}
where $\epsilon>0$ is arbitrary. From Eqs.~\ref{eq:proof-lhs}, \ref{eq:proof-rhs} we deduce then that there exists $C>0$ such that
\begin{align*}
	\frac{1}{2}\frac{d}{dt} &\left(\wnorm{(\phi_2-\phi_1)(t)}_0^2+\norm{(c_2-c_1)(t)}_0^2\right) \\
	& +(\kappa_m-\epsilon C)\intl_\Omega\left(\Phi(\phi_2)-\Phi(\phi_1)\right)(\phi_2-\phi_1)\,dx
		+D\norm{\nabla{(c_2-c_1)}(t)}_0^2 \\
	& +\eta\norm{(c_2-c_1)(t)}_{0,\partial_b\Omega}^2\leq
		C\left(\frac{1}{2}+\frac{1}{\epsilon}\right)\left(\wnorm{(\phi_2-\phi_1)(t)}_0^2+\norm{(c_2-c_1)(t)}_0^2\right)
\end{align*}
for a.e. $t\in[0,\,\Tmax]$. Particularly, it is possible to choose $\epsilon$ so small that $\kappa_m-\epsilon C>0$. Multiplying both sides by $e^{-2C't}$, $C':=C(1/2+1/\epsilon)$, yields
\begin{align*}
	\frac{1}{2}\frac{d}{dt} &\left(e^{-2C't}\left(\wnorm{(\phi_2-\phi_1)(t)}_0^2+\norm{(c_2-c_1)(t)}_0^2\right)\right) \\
	& +e^{-2C't}(\kappa_m-\epsilon C)\intl_\Omega\left(\Phi(\phi_2)-\Phi(\phi_1)\right)(\phi_2-\phi_1)\,dx \\
	& +e^{-2C't}\left(D\norm{\nabla{(c_2-c_1)}(t)}_0^2+\eta\norm{(c_2-c_1)(t)}_{0,\partial_b\Omega}^2\right)\leq 0
\end{align*}
whence, integrating from $0$ to $t\leq\Tmax$ and considering that $e^{2C't}\leq e^{2C'\Tmax}$ for all $0\leq t\leq\Tmax$,
\begin{align*}
	\frac{1}{2}\bigl(\wnorm{(\phi_2-\phi_1)(t)}_0^2 &+ \norm{(c_2-c_1)(t)}_0^2\bigr) \\
	& +(\kappa_m-\epsilon C)\intl_0^t\intl_\Omega\left(\Phi(\phi_2)-\Phi(\phi_1)\right)(\phi_2-\phi_1)\,dx\,d\tau \\
	& +2D\intl_0^t\norm{\nabla{(c_2-c_1)}(\tau)}_0^2\,d\tau+2\eta\intl_0^t\norm{(c_2-c_1)(\tau)}_{0,\partial_b\Omega}^2\,d\tau \\
	& \leq \frac{1}{2}e^{2C'\Tmax}\left(\wnorm{\phi_{2,0}-\phi_{1,0}}_0^2+\norm{c_{2,0}-c_{1,0}}_0^2\right)
\end{align*}
for all $t\in[0,\,\Tmax]$. At this point it suffices to observe that each term at the left-hand side, being nonnegative, is singularly bounded from above by the right-hand side. Integrating the first twos on $[0,\,\Tmax]$ and evaluating the remaining ones for $t=\Tmax$ gives the thesis.
\end{proof}
	
\section{The stationary problem}
\label{sect:time.independent}
In this section we turn our attention to the stationary problem
\begin{equation}
	\left\{
	\begin{array}{rcll}
		-\kappa_m\Delta{\Phi(\phi)} & = & \Gamma(x,\,\phi,\,c) & \text{in\ } \Omega_T,\,\Omega_H \\[0.15cm]
		-D\Delta{c} & = & Q(x,\,\phi,\,c) & \text{in\ } \Omega \\[0.2cm]
		\begin{array}{rcl}
			\kappa_m\nabla\Phi(\phi)\cdot\n & = & 0, \\
			\phi & = & \phi^\star, \\
			\kappa_m\jump{\nabla{\Phi(\phi)}}\cdot\n & = & 0 \\
		\end{array} & &
		\begin{array}{rcl}
			-D\nabla{c}\cdot\n & = & \eta(c-c_b) \\
			c & = & c_b \\
			\\
		\end{array} & 
		\begin{array}{l}
			\text{on\ } \partial_b\Omega \\
 			\text{on\ } \partial_f\Omega \\
			\text{on\ } S \\
		\end{array}
	\end{array}
	\right.
	\label{eq:stat}
\end{equation}
which describes the equilibrium configurations of the model for large times. The asymptotic stability of constant steady states for a reaction-diffusion system sharing some analogies with Problem \ref{eq:stat} has been addressed in \cite{difrancesco2011asc}.

Heuristically, the solution of Problem~\ref{eq:stat} is what the solution of the time-dependent Problem~\ref{eq:parab} tends to for $t\to+\infty$. Making this limit rigorous with the appropriate concept of convergence is beyond the scope of this work, therefore we will be satisfied with the above intuitive interpretation.

In Problem~\ref{eq:stat} the state variables depend on space only: $\phi=\phi(x)$, $c=c(x)$. The sub-domains $\Omega_T,\,\Omega_H$ are fixed, their interface being $S$.

We assume that solutions exist to Problem \ref{eq:stat} in the following sense:
\begin{definition}[Weak solutions to the stationary problem] \label{def:weaksol.stat}
A weak solution to Problem~\ref{eq:stat} is a pair $(\phi,\,c)\in\V$ such that:
\begin{enumerate}
\item[(i)] $\Phi(\phi)\in H^1(\Omega)$
\item[(ii)] $\phi=\phith$, $c=c_b$ on $\partial_b\Omega$ in the trace sense
\end{enumerate}
which satisfies
\begin{multline}
	\intl_\Omega\left(\kappa_m\nabla{\Phi(\phi)}\cdot\nabla{v_1}+D\nabla{c}\cdot\nabla{v_2}\right)\,dx
		+\eta\intl_{\partial_b\Omega}(c-c_b)v_2\,d\sigma \\
	=\sum_{\alpha=T,\,H}\intl_{\Omega_\alpha}\left(\sum_{\nu=p,\,d}\gamma_\alpha^\nu f_\alpha^\nu(\phi)
		g_\alpha^\nu(c)v_1-\delta\phi v_1-\lambda_\alpha h_\alpha(\phi)q_\alpha(c)v_2\right)\,dx
	\label{eq:stat-weak}
\end{multline}
for all $v_1,\,v_2\in V_f$.
\end{definition}

\subsection{Nonnegativity and boundedness of the solution}
\begin{theorem}
Any weak solution $(\phi,\,c)\in\V$ to Problem~\ref{eq:stat} satisfies
\begin{equation*}
	0\leq\phi(x)\leq\phimax, \quad 0\leq c(x)\leq c_b \quad \text{for a.e.\ } x\in\Omega.
\end{equation*}
\label{theo:boundedness-stat}
\end{theorem}
\begin{proof}
In order to prove that $\phi,\,c\geq 0$ a.e. in $\Omega$ we choose $v_1=\phi^-$, $v_2=c^-$ in Eq.~\ref{eq:stat-weak} and, similarly to the proof of Theorem \ref{theo:boundedness}, we compute:
\begin{align*}
	& -\intl_\Omega\left(\kappa_m\Phi'(-\phi^-)\vert\nabla{\phi^-}\vert^2+
		D\vert\nabla{c^-}\vert^2\right)\,dx+\eta\intl_{\partial_b\Omega}(c-c_b)c^-\,d\sigma \\
	& =\!\!\!\sum_{\alpha=T,\,H}\intl_{\Omega_\alpha}\!\!\left(\sum_{\nu=p,\,d}\gamma_\alpha^\nu f_\alpha^\nu(-\phi^-)
		g_\alpha^\nu(c)\phi^{-}+\delta(\phi^-)^2
			-\lambda_\alpha h_\alpha(\phi)q_\alpha(c^-)c^-\right)\! dx.
\end{align*}
The right-hand side being nonnegative in view of hypotheses (H2)--(H7), it results
\begin{equation*}
	-\kappa_m\intl_\Omega\Phi'(-\phi^-)\abs{\nabla{\phi^-}}^2\,dx-D\norm{\nabla{c^-}}_0^2
		-\eta\norm{c^-}_{0,\partial_b\Omega}^2-\eta c_b\intl_{\partial_b\Omega}c^-\,d\sigma\geq 0,
\end{equation*}
whence, considering that each term at the left-hand side is nonpositive, we deduce immediately $c^-=0$ a.e. in $\Omega$, and obtain furthermore $\Phi'(-\phi^-)\abs{\nabla{\phi^-}}^2=0$ a.e. in $\Omega$. This means either $\Phi'(-\phi^-)=0$, which implies $\phi^-=0$ a.e. in $\Omega$ because $\Phi'$ vanishes at most in zero, or $\vert\nabla{\phi^-}\vert^2=0$, which yields $\phi^-=0$ a.e. in $\Omega$ as well due to $\phi^-\in V_f$ (use Poincar\'e's inequality).

\bigskip

Next we prove that $\phi\leq\phimax$, $c\leq c_b$ a.e. in $\Omega$. Set $\tilde{\phi}:=(\phi-\phimax)^+$, $\tilde{c}:=(c-c_b)^+$ and choose $v_1=\tilde{\phi}$, $v_2=\tilde{c}$ as test functions in Eq.~\ref{eq:stat-weak} to find
\begin{align*}
	& \kappa_m\intl_\Omega\Phi'(\phimax+\tilde{\phi})\abs{\nabla{\tilde{\phi}}}^2\,dx
		+D\norm{\nabla{\tilde{c}}}_0^2+\eta\norm{\tilde{c}}_{0,\partial_b\Omega}^2 \\
	& =\sum_{\alpha=T,\,H}\intl_{\Omega_\alpha}\Biggl(\sum_{\nu=p,\,d}\gamma_\alpha^\nu
		f_\alpha^\nu(\phimax+\tilde{\phi})g_\alpha^\nu(c)\tilde{\phi} \\
	& \phantom{=\sum_{\alpha=T,\,H}\intl_{\Omega_\alpha}\Biggl(}
		-\delta(\phimax+\tilde{\phi})\tilde{\phi}-\lambda_\alpha h_\alpha(\phi)q_\alpha(c_b+\tilde{c})\tilde{c}\Biggr)\,dx.
\end{align*}
Arguing like in Theorem \ref{theo:boundedness}, we conclude that the right-hand side of the above equation is nonpositive, whereas each term at the left-hand side is nonnegative. Therefore we have $\tilde{c}=0$ and also $\Phi'(\phimax+\tilde{\phi})\abs{\nabla{\tilde{\phi}}}^2=0$ a.e. in $\Omega$. Since $\Phi'(\phimax+\tilde{\phi})>0$, this implies $\abs{\nabla{\tilde{\phi}}}^2=0$ a.e. in $\Omega$, i.e., $\tilde{\phi}=0$ a.e. in $\Omega$ as well thanks to $\tilde{\phi}\in V_f$.
\end{proof}

\subsection{Uniqueness of the solution}
\begin{theorem}
There exists a constant $C>0$, depending only on the coefficients $\gamma_\alpha^\nu$, $\lambda_\alpha$ and on the functions $f_\alpha^\nu$, $g_\alpha^\nu$, $h_\alpha$, $q_\alpha$, such that if $C$ is sufficiently small then Problem~\ref{eq:stat} admits at most one weak solution $(\phi,\,c)\in\V$.
\label{theo:uniqueness-stat}
\end{theorem}
\begin{proof}
Let $(\phi_i,\,c_i)\in\V$, $i=1,\,2$, be two solutions, then for all $v_1,\,v_2\in V_f$ the difference $(\phi_2-\phi_1,\,c_2-c_1)$ satisfies
\begin{align*}
	& \intl_\Omega\left(\kappa_m\nabla{(\Phi(\phi_2)-\Phi(\phi_1))}
		\cdot\nabla{v_1}+D\nabla{(c_2-c_1)}\cdot\nabla{v_2}\right)\,dx \\
	& \phantom{\intl_\Omega\left(\right.} +\eta\intl_{\partial_b\Omega}(c_2-c_1)v_2\,d\sigma \\
	& =\!\!\!\sum_{\alpha=T,\,H}\intl_{\Omega_\alpha}\!\Biggl(\sum_{\nu=p,\,d}\gamma_\alpha^\nu\left\{
		\left[f_\alpha^\nu(\phi_2)-f_\alpha^\nu(\phi_1)\right]g_\alpha^\nu(c_2)
			+f_\alpha^\nu(\phi_1)\left[g_\alpha^\nu(c_2)-g_\alpha^\nu(c_1)\right]\right\}v_1 \\
	& \phantom{=\sum_{\alpha=T,\,H}\intl_{\Omega_\alpha}\Biggl(}
		-\delta(\phi_2-\phi_1)v_1-\lambda_\alpha(h_\alpha(\phi_2)-h_\alpha(\phi_1))q_\alpha(c_2) \\
	& \phantom{=\sum_{\alpha=T,\,H}\intl_{\Omega_\alpha}\Biggl(}
		+\lambda_\alpha h_\alpha(\phi_1)(q_\alpha(c_2)-q_\alpha(c_1))v_2\Biggr)\,dx.
\end{align*}
We choose $v_1=\P(\phi_2-\phi_1)$ (cf. Appendix \ref{app:poisson}), $v_2=c_2-c_1$, and mimic the computations of the proof of Theorem \ref{theo:uniqueness} (using in particular Cauchy's inequality at the right-hand side with $\epsilon=1$) to deduce
\begin{align*}
	\kappa_m &\intl_\Omega(\Phi(\phi_2)-\Phi(\phi_1))(\phi_2-\phi_1)\,dx+D\norm{\nabla{(c_2-c_1)}}_0^2 \\
	& +\eta\norm{c_2-c_1}_{0,\partial_b\Omega}^2+\delta\wnorm{\phi_2-\phi_1}_0^2 \\
	& \lesssim \intl_\Omega\left(\Phi(\phi_2)-\Phi(\phi_1)\right)(\phi_2-\phi_1)\,dx 	
		+\wnorm{\phi_2-\phi_1}_0^2+\norm{c_2-c_1}_0^2.
\end{align*}
Hence there exists $C>0$ such that
\begin{align*}
	& (\kappa_m-C)\intl_\Omega\left(\Phi(\phi_2)-\Phi(\phi_1)\right)(\phi_2-\phi_1)\,dx
		+(\delta-C)\wnorm{\phi_2-\phi_1}_0^2 \\
	& \phantom{(\kappa} +\frac{D-CC_P^2}{1+C_P^2}\norm{c_2-c_1}_1^2+\eta\norm{c_2-c_1}_{0,\partial_b\Omega}^2\leq 0,
\end{align*}
where we further applied Poincar\'e's inequality to $c_2-c_1\in V_f$. Uniqueness of the solution follows from this relationship as long as the coefficients of all terms are positive, which entails
\begin{equation*}
	C<\min\left\{\kappa_m,\,\delta,\,\frac{D}{C_P^2}\right\}. \qedhere
\end{equation*}
\end{proof}

\begin{remark}
\label{rem:C-uniqueness-stat}
For the sake of definiteness, we record that a possible constant $C$ for Theorem \ref{theo:uniqueness-stat} is
\begin{align*}
	C=\frac{1}{2}\max\Biggl\{ &\sum_{\alpha=T,\,H}\left(\sum_{\nu=p,\,d}\gamma_\alpha^\nu\LipPhi{f_\alpha^\nu}\norm{g_\alpha^\nu}_\infty
		+\lambda_\alpha\LipPhi{h_\alpha}\norm{q_\alpha}_\infty\right), \\
	& \sum_{\alpha=T,\,H}\left(\sum_{\nu=p,\,d}\gamma_\alpha^\nu\norm{f_\alpha^\nu}_\infty\operatorname{Lip}^2(g_\alpha^\nu)
		+\lambda_\alpha\norm{q_\alpha}_\infty\right), \\
	& C_P^2\sum_{\substack{\alpha=T,\,H \\ \nu=p,\,d}}\gamma_\alpha^\nu\left(\norm{f_\alpha^\nu}_\infty
			+\norm{g_\alpha^\nu}_\infty\right)\Biggr\}.
\end{align*}
\end{remark}

\subsection{Existence of the solution}
We complete our analysis of model \ref{eq:stat} by outlining the theory of the existence of solutions. We confine ourselves to the one-dimensional setting, taking as reference domain the dimensionless interval $I=(0,\,1)$. In particular, $x=0$ will be the vascular boundary and $x=1$ the far boundary.

The case $d=1$ allows us to rely on two basic tools, which are not available in higher dimensions: on the one hand the Sobolev embedding $C^0(\bar{I})\subset H^1(I)$, on the other hand Morrey's inequality $\norm{u}_\infty\leq\norm{u}_1$ for $u\in H^1(I)$. Extending the theory to the case $d>1$ is likely to require partly different tools, which is at present beyond the scope of the work.

The one-dimensional problem is written as
\begin{equation}
	\left\{
	\begin{array}{rcll}
		-\kappa_m\Phi(\phi)_{xx} & = & \Gamma(x,\,\phi,\,c) & \text{in\ } (0,\,S),\,(S,\,1) \\[0.15cm]
		-Dc_{xx} & = & Q(x,\,\phi,\,c) & \text{in\ } I \\[0.2cm]
		\begin{array}{rcl}
			\kappa_m\Phi(\phi)_x(0) & = & 0, \\
			\phi(1) & = & \phith, \\
			\kappa_m\jump{\Phi(\phi)_x} & = & 0 \\
		\end{array} & &
		\begin{array}{rcl}
			-Dc_x(0) & = & \eta(c(0)-c_b) \\
			c(1) & = & c_b \\
			\\
		\end{array} & 
	\end{array}
	\right.
	\label{eq:stat-1d}
\end{equation}
where $S\in\bar{I}$ is the location of the point interface between tumor and host cells. In particular, $\Omega_T=(0,\,S)$ and $\Omega_H=(S,\,1)$. By adapting Definition \ref{def:weaksol.stat} to the present context, a weak solution to Problem \ref{eq:stat-1d} is a pair $(\phi,\,c)\in\V$, such that $\Phi(\phi)\in H^1(I)$, $\phi(1)=\phith$, $c(1)=c_b$, which satisfies
\begin{multline}
	\intl_0^1\left(\kappa_m\Phi(\phi)_x v_{1x}+Dc_x v_{2x}\right)\,dx+\eta\left(c(0)-c_b\right)v_2(0) \\
	=\sum_{\alpha=T,\,H}\intl_{\Omega_\alpha}\left(\sum_{\nu=p,\,d}\gamma_\alpha^\nu f_\alpha^\nu(\phi)
		g_\alpha^\nu(c)v_1-\delta\phi v_1-\lambda_\alpha h_\alpha(\phi)q_\alpha(c)v_2\right)\,dx
	\label{eq:stat-1d-weak}
\end{multline}
for all $v_1,\,v_2\in V_f$.

For the subsequent theory, it is useful to introduce the inverse $\Phi^{-1}$ of the constitutive function. Owing to hypothesis (H1), $\Phi^{-1}$ is continuous and strictly increasing, with $\Phi^{-1}(0)=0$, $\Phi^{-1}(s)<0$ for $s<0$, and $\Phi^{-1}(s)>0$ for $s>0$. It is also smooth on $(-\infty,\,0)$ and $(0,\,+\infty)$, however $\lim_{s\to 0}(\Phi^{-1})'(s)=+\infty$ because of the degeneracy of $\Phi'$ at the origin.

In order to prove the existence of stationary solutions we resort to a splitting method, which consists in approaching the two differential equations of Problem~\ref{eq:stat-1d} separately, assuming that the main unknown is either $\phi$ or $c$ and that the other function is known.

\begin{theorem}
Assume $h_\alpha(0)=0$. There exists a constant $C>0$, depending only on the parameters $\kappa_m$, $\gamma_\alpha^\nu$, $\delta$, $\phimax$ and on the functions $\Phi$, $f_\alpha^\nu$, $g_\alpha^\nu$, such that if $C$ is sufficiently small then Problem~\ref{eq:stat-1d} admits a weak solution $(\phi,\,c)\in\V$.
\label{theo:existence-stat}
\end{theorem}
\begin{proof}
We preliminarily define the sets
\begin{align*}
	\VV &:= \{f\in L^2(I)\,:\,0\leq f\leq\phimax\ \text{a.e. in\ } I\} \\
	\UU &:= \{f\in L^2(I)\,:\,0\leq f\leq c_b\ \text{a.e. in\ } I\}.
\end{align*}

\bigskip

Let us begin by considering the problem
\begin{equation}
	\left\{
	\begin{array}{rcll}
		-Dc_{xx} & = & Q(x,\,\varphi,\,c) & \text{in\ } I \\
		Dc_x(0) & = & \eta(c(0)-c_b) \\
		c(L) & = & c_b,
	\end{array}
	\right.
	\label{eq:stat-c}
\end{equation}
where $\varphi\in\VV$ is given. We associate with it an auxiliary problem in which the function $q_\alpha$ is replaced by $\tilde{q}_\alpha=q_\alpha\1_{[0,\,+\infty)}$. The corresponding weak formulation is obtained from Eq.~\ref{eq:stat-1d-weak} by letting $v_1=0$ and writing $\tilde{q}_\alpha$ in place of $q_\alpha$: find $c\in H^1(I)$, with $c(1)=c_b$, such that
\begin{equation}
	D\intl_0^1 c_x v_x\,dx+\eta(c(0)-c_b)v(0)=
		-\sum_{\alpha=T,\,H}\intl_{\Omega_\alpha}\lambda_\alpha h_\alpha(\varphi)\tilde{q}_\alpha(c)v\,dx
	\label{eq:stat-c-weak}
\end{equation}
for all $v\in V_f$. By introducing the antiderivative $\tilde{Q}_\alpha$ of $\tilde{q}_\alpha$ vanishing in zero, we can view Eq.~\ref{eq:stat-c-weak} as the Euler-Lagrange equation for the functional
\begin{equation*}
	J_1(c)=\frac{D}{2}\intl_0^1 c_x^2\,dx+\frac{\eta}{2}\left(c(0)-c_b\right)^2+
		\sum_{\alpha=T,\,H}\intl_{\Omega_\alpha}\lambda_\alpha h_\alpha(\varphi)\tilde{Q}_\alpha(c)\,dx
\end{equation*}
over the class of admissible functions $\A_1=\{c\in H^1(I):c(1)=c_b\}$. Thus we can seek our solution $c$ as a minimizing point of $J_1$ on $\A_1$.
	
Since $\tilde{q}_\alpha(s)=0$ for $s<0$ and $\tilde{q}_\alpha(s)=q_\alpha(s)\geq 0$ for $s\geq 0$, we have $\tilde{Q}_\alpha(s)\geq 0$ for all $s\in\R$. Using further that $\lambda_\alpha$ and $h_\alpha$ are nonnegative we obtain $J_1(c)\geq D/2\norm{c_x}_0^2$, which implies that $J_1$ is coercive. Therefore any minimizing sequence $\{c_k\}_{k=1}^{\infty}\subseteq\A_1$ is bounded in $H^1(I)$ and, upon passing to a subsequence, we can assume that it converges weakly to some $\bar{c}\in H^1(0,\,L)$. But $c_k-c_b\in V_f$ all $k$ and $V_f$ is a weakly closed subspace of $H^1(I)$ (in view of Mazur's Theorem, as it is closed), thus we deduce more precisely $\bar{c}-c_b\in V_f$, i.e., $\bar{c}(1)=c_b$ and ultimately $\bar{c}\in\A_1$.

Considering that $J_1$ is of the form $\int_0^1\Lag_1(c_x,\,c,\,x)\,dx$ for the Lagrangian
\begin{equation*}
	\Lag_1(p,\,z,\,x)=\frac{D}{2}p^2-\eta(z-c_b)p+
		\sum_{\alpha=T,\,H}\lambda_\alpha(x)h_\alpha(\varphi(x))\tilde{Q}_\alpha(z)\1_{\Omega_\alpha}(x),
\end{equation*}
which is smooth and convex in $p$ for all $z\in\R$ and all $x\in I$, we deduce that $J_1$ is sequentially weakly lower semicontinuous on $H^1(I)$. Thus $\bar{c}$ is a minimizing point of $J_1$, i.e., a solution to our auxiliary problem.

Mimicking the computations of the proof of Theorem \ref{theo:boundedness-stat} with $v_1=0$ reveals that, for any fixed $\varphi\in\VV$, all solutions to the auxiliary problem range in $[0,\,c_b]$. Hence we conclude $0\leq\bar{c}(x)\leq c_b$ for all $x\in\bar{I}$, and consequently that $\bar{c}$ solves also Problem~\ref{eq:stat-c} because the latter and the auxiliary problem coincide for $c\in[0,\,c_b]$. Notice that $\bar{c}\in\UU$.

We show now that the solutions to Problem~\ref{eq:stat-c} depend continuously on $\varphi$. Let $c_1,\,c_2$ be two solutions corresponding to $\varphi_1,\,\varphi_2\in\VV$, respectively, then for all $v\in V_f$ the difference $c_2-c_1$ satisfies
\begin{multline*}
	D\intl_0^1(c_2-c_1)_x v_x\,dx+\eta\left(c_2(0)-c_1(0)\right)v(0) \\
	=-\sum_{\alpha=T,\,H}\intl_{\Omega_\alpha}\lambda_\alpha\{(h_\alpha(\varphi_2)-h_\alpha(\varphi_1))q_\alpha(c_2)+
		h_\alpha(\varphi_1)(q_\alpha(c_2)-q_\alpha(c_1))\}v\,dx.
\end{multline*}			
We choose $v=c_2-c_1$ and observe that $-\lambda_\alpha h_\alpha(\varphi_1)(q_\alpha(c_2)-q_\alpha(c_1))(c_2-c_1)\leq 0$ in $\Omega_\alpha$ (hypothesis (H7.2)), whence
\begin{align*}
	D\norm{(c_2-c_1)_x}_0^2 &+ \eta\left(c_2(0)-c_1(0)\right)^2 \\
	& \leq\sum_{\alpha=T,\,H}\intl_{\Omega_\alpha}\lambda_\alpha q_\alpha(c_2)
		\abs{h_\alpha(\varphi_2)-h_\alpha(\varphi_1)}\cdot\abs{c_2-c_1}\,dx \\
	& \leq\frac{1}{2}\sum_{\alpha=T,\,H}\lambda_\alpha\|q_\alpha\|_\infty
		\left(\frac{1}{\epsilon}\norm{h_\alpha(\varphi_2)-h_\alpha(\varphi_1)}_0^2+\epsilon\norm{c_2-c_1}_0^2\right).
\end{align*}
Now we recall, from hypothesis (H6), that $h_\alpha$ is $\Phi$-Lipschitz continuous in $[0,\,\phi_\ast]$, whence
\begin{align*}
	\norm{h_\alpha(\varphi_2)-h_\alpha(\varphi_1)}_0^2 &\leq
		\LipPhi{h_\alpha}\intl_0^1\left(\Phi(\varphi_2)-\Phi(\varphi_1)\right)(\varphi_2-\varphi_1)\,dx \\
	& \lesssim\norm{\varphi_2-\varphi_1}_0^2.
\end{align*}
Thus the previous computation can be continued by asserting that there exists $C>0$ such that
\begin{equation}
	\frac{D-\epsilon CC_P^2}{1+C_P^2}\norm{c_2-c_1}_1^2+\eta\left(c_2(0)-c_1(0)\right)^2\leq
		\frac{C}{\epsilon}\norm{\varphi_2-\varphi_1}_0^2,
	\label{eq:stat-contdep-c}
\end{equation}
where we also applied Poincar\'e's inequality to $c_2-c_1\in V_f$. Choosing $\epsilon>0$ so small that $D-\epsilon CC_P^2>0$, we get from Eq.~\ref{eq:stat-contdep-c} the desired continuity estimate. In particular, for $\varphi_1=\varphi_2$ we obtain the uniqueness of the solution to Problem~\ref{eq:stat-c}.

The foregoing results enable us to define the operator $\S_1:\VV\to\UU$ such that $\S_1(\varphi)=c$. From the continuity estimate \ref{eq:stat-contdep-c} we deduce that $\S_1$ is Lipschitz continuous on $\VV$ and, with a little more work, that it is also compact. To see this, we observe first of all that the assumption $h_\alpha(0)=0$ implies $\S_1(0)=c_b$ (i.e., the unique solution to Problem~\ref{eq:stat-c} for $\varphi=0$ is $c=c_b$), then we choose $\varphi_1=0$ in Eq.~\ref{eq:stat-contdep-c} and drop the subindex $2$ to obtain
\begin{equation*}
	\norm{\S_1(\varphi)-c_b}_1^2\lesssim\norm{\varphi}_0^2.
\end{equation*}
We take now $\{\varphi_k\}_{k=1}^{\infty}\subseteq\VV$ and notice that, in view of the latter estimate, the sequence $\{\S_1(\varphi_k)-c_b\}_{k=1}^{\infty}$ is bounded in $H^1(I)$. Owing to Rellich's Theorem, we can therefore assume, upon passing to a subsequence, that $\S_1(\varphi_k)-c_b$ converges in $L^2(I)$ as $k\to\infty$, i.e., that the sequence $\{\S_1(\varphi_k)\}_{k=1}^{\infty}\subseteq\UU$ is convergent, which proves the compactness of $\S_1$.

\bigskip

We turn now our attention to the problem
\begin{equation}
	\left\{
	\begin{array}{rcll}
		-\kappa_m\Phi(\phi)_{xx} & = & \Gamma(x,\,\phi,\,\theta) & \text{in\ } (0,\,S),\,(S,\,1) \\
		\kappa_m\Phi(\phi)_x(0) & = & 0 \\
		\phi(1) & = & \phith \\
		\kappa_m\jump{\Phi(\phi)_x} & = & 0,
	\end{array}
	\right.
	\label{eq:stat-phi}
\end{equation}
where $\theta\in\UU$ is given. Again, we associate with it an auxiliary problem in which the functions $f_\alpha^p$, $f_\alpha^d$ are replaced by $\tilde{f}_\alpha^p=f_\alpha^p\1_{[0,\,\phimax]}$, $\tilde{f}_\alpha^d=f_\alpha^d\1_{[0,\,+\infty)}$, respectively. The weak formulation is recovered from Eq.~\ref{eq:stat-1d-weak} by letting $v_2=0$ and substituting conveniently the functions at the right-hand side: find $\phi\in H^1(I)$, with $\Phi(\phi)\in H^1(I)$ and $\phi(1)=\phith$, such that
\begin{equation*}
	\kappa_m\intl_0^1\Phi(\phi)_x v_x\,dx=
		\sum_{\alpha=T,\,H}\intl_{\Omega_\alpha}\left(\sum_{\nu=p,\,d}\gamma_\alpha^\nu\tilde{f}_\alpha^\nu(\phi)
			g_\alpha^\nu(\theta)-\delta\phi\right)v\,dx
\end{equation*}
for all $v\in V_f$. We set $u:=\Phi(\phi)$, whence $\phi=\Phi^{-1}(u)$, so that this equation becomes
\begin{equation}
	\kappa_m\intl_0^1 u_x v_x\,dx=\!\!\!
		\sum_{\alpha=T,\,H}\intl_{\Omega_\alpha}\!\!\left(\sum_{\nu=p,\,d}\gamma_\alpha^\nu(\tilde{f}_\alpha^\nu\circ\Phi^{-1})(u)
			g_\alpha^\nu(\theta)-\delta\Phi^{-1}(u)\right)\!\! v\,dx
	\label{eq:stat-phi-weak}
\end{equation}
for all $v\in V_f$. If we introduce the antiderivatives $\tilde{F}_\alpha^\nu$, $\Psi$ of $\tilde{f}_\alpha^\nu\circ\Phi^{-1}$, $\Phi^{-1}$ vanishing in zero, we can regard Eq.~\ref{eq:stat-phi-weak} as the Euler-Lagrange equation for the functional
\begin{equation*}
	J_2(u)=\frac{\kappa_m}{2}\intl_0^1 u_x^2\,dx-
		\sum_{\alpha=T,\,H}\intl_{\Omega_\alpha}\left(\sum_{\nu=p,\,d}\gamma_\alpha^\nu\tilde{F}_\alpha^\nu(u)
			g_\alpha^\nu(\theta)-\delta\Psi(u)\right)\,dx
\end{equation*}
over the class of admissible functions $\A_2=\{u\in H^1(I):u(1)=\Phi(\phith)\}$. Thus, again we can look for our solution $u$ as a minimizing point of $J_2$ on $\A_2$.

Notice that $\tilde{F}_\alpha^p(s)\leq\tilde{F}_\alpha^p(\Phi(\phimax))$, and that $\tilde{F}_\alpha^d(s),\,\Psi(s)\geq 0$ for all $s\in\R$. Therefore, recalling further that $\gamma_\alpha^d<0$ (hypothesis (H2)), we have
\begin{equation*}
	J_2(u)\geq\frac{\kappa_m}{2}\norm{u_x}_0^2
		-\sum_{\alpha=T,\,H}\gamma_\alpha^p\norm{g_\alpha^p}_\infty\tilde{F}_\alpha^p(\Phi(\phimax))\abs{\Omega_\alpha}.
\end{equation*}
$J_2$ is thus coercive, hence any minimizing sequence $\{u_k\}_{k=1}^{\infty}\subseteq\A_2$ converges weakly (up to possibly passing to subsequences) to some $\bar{u}\in H^1(I)$. Since $u_k-\Phi(\phith)\in V_f$ and $V_f$ is weakly closed in $H^1(I)$, it results $\bar{u}-\Phi(\phith)\in V_f$, that is $\bar{u}\in\A_2$. In addition, $J_2$ is in turn of the form $\int_0^1\Lag_2(u_x,\,u,\,x)\,dx$ for the Lagrangian
\begin{equation*}
	\Lag_2(p,\,z,\,x)=\frac{\kappa_m}{2}p^2-\sum_{\alpha=T,\,H}\left(\sum_{\nu=p,\,d}\gamma_\alpha^\nu
		\tilde{F}_\alpha^\nu(z)g_\alpha^\nu(\theta(x))-\delta\Psi(z)\right)\1_{\Omega_\alpha}(x),
\end{equation*}
which is smooth and convex in $p$ for each $z\in\R$, $x\in I$, hence $J_2$ is sequentially weakly lower semicontinuous on $H^1(I)$. It follows that $\bar{u}$ is a minimizing point for $J_2$ on $\A_2$, and consequently $\bar{\phi}:=\Phi^{-1}(\bar{u})$ is a weak solution to our auxiliary problem.

Mimic now the computations of the proof of Theorem \ref{theo:boundedness-stat} with $v_2=0$ to obtain that, for any fixed $\theta\in\UU$, all solutions to the auxiliary problem range in $[0,\,\phimax]$, whence $0\leq\bar{\phi}\leq\phimax$ in $\bar{I}$. But the auxiliary problem and Problem~\ref{eq:stat-phi} coincide for $\phi\in[0,\,\phimax]$, hence ultimately we have found a solution $\bar{\phi}\in\VV$ to Problem~\ref{eq:stat-phi}.

Next we show that, by introducing suitable constraints on the parameters, we can guarantee that $\bar{\phi}$ be strictly positive in $\bar{I}$. Let us pick $v_1=\Phi(\phi)-\Phi(\phith)$, $v_2=0$ as test functions in Eq.~\ref{eq:stat-1d-weak} to discover
\begin{align*}
	\kappa_m\intl_0^1\Phi(\phi)_x\left(\Phi(\phi)-\Phi(\phith)\right)_x\,dx &=
		\sum_{\substack{\alpha=T,\,H\\ \nu=p,\,d}}\intl_{\Omega_\alpha}\gamma_\alpha^\nu f_\alpha^\nu(\phi)
			g_\alpha^\nu(\theta)\left(\Phi(\phi)-\Phi(\phith)\right)\,dx \\
	&\phantom{=}-\delta\intl_0^1\phi\left(\Phi(\phi)-\Phi(\phith)\right)\,dx.
\end{align*}
Noting that $\Phi(\phi)_x=\left(\Phi(\phi)-\Phi(\phith)\right)_x$ at the left-hand side and using the boundedness of $f_\alpha^\nu$, $g_\alpha^\nu$, $\phi$ at the right-hand side, we estimate
\begin{equation*}
	\norm{(\Phi(\phi)-\Phi(\phith))_x}_0^2\lesssim\norm{\Phi(\phi)-\Phi(\phith)}_\infty.
\end{equation*}
In addition, owing to Poincar\'e's and Morrey's inequalities,
\begin{equation*}
	\norm{(\Phi(\phi)-\Phi(\phith))_x}_0^2\gtrsim\norm{\Phi(\phi)-\Phi(\phith)}_1^2\gtrsim\norm{\Phi(\phi)-\Phi(\phith)}_\infty^2,
\end{equation*}
hence finally there exists $C_1>0$ such that $\norm{\Phi(\phi)-\Phi(\phith)}_\infty\leq C_1$, which indicates that $\Phi(\phi(x))\geq\Phi(\phith)-C_1$ for all $x\in\bar{I}$. For definiteness, we report the explicit expression of a possible constant $C_1$:
\begin{equation*}
	C_1=\frac{1+C_P^2}{\kappa_m}\left(\sum_{\substack{\alpha=T,\,H \\ \nu=p,\,d}}\gamma_\alpha^\nu\norm{f_\alpha^\nu}_\infty
		\norm{g_\alpha^\nu}_\infty\abs{\Omega_\alpha}+\delta\phimax\right).
\end{equation*}
We fix now $\epsilon\in(0,\,\phith)$ and observe that $\phi\geq\epsilon$ if and only if $\Phi(\phi)\geq\Phi(\epsilon)$, thus we can guarantee that $\phi$ be strictly positive in $\bar{I}$ if we require $\Phi(\phith)-C_1\geq\Phi(\epsilon)$, which implies the constraint
\begin{equation}
	C_1\leq\Phi(\phith)-\Phi(\epsilon).
	\label{eq:constr1}
\end{equation}
Given this, any solution $\phi\in\VV$ to Problem~\ref{eq:stat-phi} satisfies $0<\epsilon\leq\phi\leq\phimax$ in $\bar{I}$.

Finally we assert that, under condition \ref{eq:constr1}, solutions to Problem~\ref{eq:stat-phi} depend continuously on $\theta\in\UU$ in the norm $\norm{\cdot}_0$. For this, let $\phi_1,\,\phi_2$ be two solutions corresponding to $\theta_1,\,\theta_2\in\UU$, respectively, then for all $v\in V_f$ their difference $\phi_2-\phi_1$ solves
\begin{multline*}
	\kappa_m\intl_0^1(\Phi(\phi_2)-\Phi(\phi_1))_x v_x\,dx+\delta\intl_0^1(\phi_2-\phi_1)v\,dx \\
	=\sum_{\substack{\alpha=T,\,H \\ \nu=p,\,d}}
		\intl_{\Omega_\alpha}\gamma_\alpha^\nu\{(f_\alpha^\nu(\phi_2)-f_\alpha^\nu(\phi_1))g_\alpha^\nu(\theta_2)
			+f_\alpha^\nu(\phi_1)(g_\alpha^\nu(\theta_2)-g_\alpha^\nu(\theta_1))\}v\,dx.
\end{multline*}
We choose $v=\P(\phi_2-\phi_1)$ (cf. Appendix \ref{app:poisson}) and, mimicking the computations of Theorem \ref{theo:uniqueness-stat}, we find that there exists $C_2>0$ such that
\begin{equation*}
	(\kappa_m-C_2)\intl_0^1\left(\Phi(\phi_2)-\Phi(\phi_1)\right)(\phi_2-\phi_1)\,dx+(\delta-C_2)\wnorm{\phi_2-\phi_1}_0^2
		\lesssim\norm{\theta_2-\theta_1}_0^2
\end{equation*}
(for the sake of completeness, we point out that the constant $C_2$ is the same as the one appearing in Theorem \ref{theo:uniqueness-stat}, cf. also Remark \ref{rem:C-uniqueness-stat}). Assume
\begin{equation*}
	C_2\leq\min\{\kappa_m,\,\delta\},
\end{equation*}
then, since $\abs{\Phi(\phi_2)-\Phi(\phi_1)}\geq\left(\min_{s\in[\epsilon,\,\phimax]}\Phi'(s)\right)\abs{\phi_2-\phi_1}$, it follows
\begin{equation}
	\min_{s\in[\epsilon,\,\phimax]}\Phi'(s)(\kappa_m-C_2)\norm{\phi_2-\phi_1}_0^2+
		(\delta-C_2)\wnorm{\phi_2-\phi_1}_0^2\lesssim\norm{\theta_2-\theta_1}_0^2,
	\label{eq:cont_est-phi}
\end{equation}
which yields the desired continuity estimate, together with uniqueness of the solution to Problem~\ref{eq:stat-phi} when $\theta_1=\theta_2$.

Define now $C:=\max\{C_1,\,C_2\}$ and impose
\begin{equation*}
	C<\min\{\Phi(\phith)-\Phi(\epsilon),\,\kappa_m,\,\delta\},
\end{equation*}
then Problem~\ref{eq:stat-phi} admits a unique solution $\phi\in\VV$ for any given $\theta\in\UU$. Consequently, we are in a position to define the operator $\S_2:\UU\to\VV$ such that $\S_2(\theta)=\phi$, which is Lipschitz continuous on $\UU$ in view of Eq.~\ref{eq:cont_est-phi}.

\bigskip

At last, we come back to the full Problem~\ref{eq:stat-1d} in this way: we construct by composition the operator $\S:=\S_2\circ\S_1:\VV\to\VV$ such that $\S(\varphi)=\phi$. Since $\S_1$ is continuous and compact and $\S_2$ is continuous, $\S$ is in turn continuous and compact; moreover, $\VV$ is convex and closed in $L^2(I)$ (to see the latter property, use that convergence in $L^2(I)$ implies pointwise convergence a.e. in $I$ upon passing to subsequences). Schauder's Fixed Point Theorem implies then that $\S$ has a fixed point $\phi\in\VV$, hence the pair $(\phi,\,c=\S_1(\phi))$ is a weak solution to Problem~\ref{eq:stat-1d} and we are done.
\end{proof}

\begin{remark}
If, in addition to the hypotheses of Theorem \ref{theo:existence-stat}, also the hypotheses of Theorem \ref{theo:uniqueness-stat} hold true then the solution to Problem~\ref{eq:stat-1d} is unique.
\end{remark}

\section{Possible developments}
\label{sect:developments}
In this paper we have addressed the mathematical formulation of initial and boundary-value problems for multiphase models of tumor growth, deduced from the framework developed in \cite{MR2471305}. We have performed a qualitative analysis of both the time-dependent and the time-independent problems, mainly by means of $L^2$-$H^1$ \emph{a priori} estimates, establishing nonnegativity, boundedness, uniqueness, and continuous dependence of the solution on the initial data. In the one-dimensional time-independent case we have also obtained the existence of the solution.

The analytical techniques used here may be profitably exploited to approach more advanced multiphase models, also fitting the framework presented in \cite{MR2471305}, which incorporate a more accurate description of the interactions between the cells and the extracellular matrix. Based on phenomenological laboratory observations \cite{baumgartner2000cip,canetta2005mcv,sun2005mmt}, they take into account the adhesion of the former to the latter by relating the cell-matrix stress $\m_{\alpha m}$ (a component of the overall external stress $\m_\alpha$ included in Eq.~\ref{eq:mixture-stress}) to the cell-matrix relative velocity $\v_m-\v_\alpha$ as
\begin{equation}
	\v_m-\v_\alpha=\K_{\alpha m}\left(1-\frac{\Sigma_{\alpha m}}{\vert\m_{\alpha m}\vert}\right)^+\m_{\alpha m}.
	\label{eq:mam}
\end{equation}
This formula says that if the magnitude of the stress $\m_{\alpha m}$ is below some critical threshold $\Sigma_{\alpha m}>0$ then there is no relative motion between the cells and the matrix, that is the former remain anchored to the latter. Conversely, if $\vert\m_{\alpha m}\vert$ is above the threshold $\Sigma_{\alpha m}$ then the interaction stress $\m_{\alpha m}$ is proportional to the relative velocity $\v_m-\v_\alpha$, thus recovering a more classical viscous friction which, in particular, means that cells slide on the matrix with their own velocity. If Eq.~\ref{eq:mam} is used, with the additional assumption of motionless matrix ($\v_m=0$), then the equations ruling cell dynamics take the form
\begin{equation}
	\frac{\partial\phi_\alpha}{\partial t}-\nabla\cdot[\phi_\alpha
		\I_\alpha(\phi_T,\,\phi_H,\,\vert\nabla{\phi}\vert)\K_{\alpha m}\nabla{(\phi\Sigma(\phi)})]
			=\Gamma_\alpha,
	\label{eq:phi-attach_detach}
\end{equation}
where
\begin{equation*}
	\I_\alpha(\phi_T,\,\phi_H,\,\vert\nabla{\phi}\vert)=
		{\left(\frac{\phi_\alpha}{\phi}-\frac{\Sigma_{\alpha m}}{\vert\nabla(\phi\Sigma(\phi))\vert}\right)}^+
\end{equation*}
translates the adhesion mechanisms discussed above. In particular, the velocity $\v_\alpha$ of the cells is 
\begin{equation*}
	\v_\alpha=-\I_\alpha\K_{\alpha m}\nabla(\phi\Sigma(\phi)),
\end{equation*}
hence if $\vert\nabla(\phi\Sigma(\phi))\vert<\Sigma_{\alpha m}$ then $\v_\alpha=0$ because $\I_\alpha=0$ (recall that $\phi_\alpha\leq\phi$ by definition) and the cells stay attached to the matrix, while if $\vert\nabla(\phi\Sigma(\phi))\vert>\Sigma_{\alpha m}$ the cells might detach from the matrix since one may have $\I_\alpha>0$.

Equation~\ref{eq:phi-attach_detach} can be regarded as a refined version of Eq.~\ref{eq:mixture-general}, which would be interesting to study in view of its physical significance, possibly adapting the techniques illustrated in this paper. Notice indeed that setting $\Sigma_{\alpha m}=0$ for both $\alpha=T$ and $\alpha=H$, which amounts to assuming a purely viscous friction between the cells and the matrix without attachment/detachment, reduces Eq.~\ref{eq:phi-attach_detach} to Eq.~\ref{eq:mixture-general}, hence the latter turns out to be a particular case of the former. Additional mathematical difficulties need however to be overcome, especially the harder degeneracy of the differential operator in space caused by the term $\I_\alpha$.

\section*{Acknowledgements}
The author wants to address special thanks to prof. Luigi Preziosi, who inspired this research line and kindly proofread the final version of the manuscript.

\appendix

\section{The Poisson solution operator}
\label{app:poisson}
In this appendix we introduce a basic tool, inspired by \cite{fadimba1995ape,laurencot2005cmt}, useful to handle the nonlinearity $\Phi$ of Eq.~\ref{eq:nonlin.diff}.

Consider the linear elliptic problem
\begin{equation}
	\left\{
	\begin{array}{rcll}
		-\Delta{u} & = & f & \text{in\ } \Omega \\[0.1cm]
		\nabla{u}\cdot\n & = & 0 & \text{on\ } \partial_b\Omega \\[0.1cm]
		u & = & 0 & \text{on\ } \partial_f\Omega
	\end{array}
	\right.
	\label{eq:Poisson}
\end{equation}
and assume $\partial_f\Omega\ne\emptyset$ (see Appendix~\ref{app:no.far} for the case $\partial_f\Omega=\emptyset$), $f\in L^2(\Omega)$. The classical weak formulation is: 
\begin{equation}
	\left\{
	\begin{array}{c}
		\text{find $u\in V_f$ such that} \\[0.2cm]
		\displaystyle{\intl_{\Omega}}\nabla{u}\cdot\nabla{v}\,dx=\displaystyle{\intl_{\Omega}}fv\,dx, \quad \forall\,v\in V_f,
	\end{array}
	\right.
	\label{eq:Poisson-weak}
\end{equation}
which, owing to Lax-Milgram theory, yields a unique weak solution $u$ with $\norm{\nabla{u}}_0\lesssim\norm{f}_0$.

Let us introduce the linear operator $\P:L^2(\Omega)\to L^2(\Omega)$, termed the Poisson solution operator, which associates with $f$ the solution $u=\P{f}$ to Problem~\ref{eq:Poisson-weak} above. $\P$ is bounded, symmetric, and positive definite. Therefore, besides the standard inner product $\innprod{f}{g}=\int_\Omega f(x)g(x)\,dx$, we can endow $L^2(\Omega)$ with an inner product induced by $\P$:
$\winnprod{f}{g}:=\innprod{\P{f}}{g}$, with corresponding norm $\wnorm{f}_0^2:=\winnprod{f}{f}=\innprod{\P{f}}{f}$. Using $\P$ we rewrite Eq.~\ref{eq:Poisson-weak} as
\begin{equation}
	\intl_{\Omega}\nabla{\P{f}}\cdot\nabla{v}\,dx = \intl_{\Omega}fv\,dx, \quad \forall\,f\in L^2(\Omega),\ v\in V_f,
	\label{eq:action_of_P}
\end{equation}
whence, letting $v=\P{f}\in V_f$, we see that $\wnorm{f}_0=\norm{\nabla{\P{f}}}_0$. Thus the previous \emph{a priori} estimate on $u$ implies $\wnorm{f}_0\lesssim\norm{f}_0$. On the other hand, in view of Poincar\'e's inequality, we also have $\norm{\P{f}}_0\lesssim \norm{\nabla{\P{f}}}_0=\wnorm{f}_0$.

Given $u\in L^2(0,\,\Tmax;\,V_f)$ with $u_t\in L^2(0,\,\Tmax;\,V_f')$, the duality pairing between $u_t(t)\in V_f'$ and $\P{u(t)}\in V_f$ yields, for a.e. $t$,
\begin{equation*}
	\dual{u_t(t)}{\P{u(t)}}=\frac{1}{2}\frac{d}{dt}\wnorm{u(t)}_0^2.
\end{equation*}

\section{Problems with no far boundary}
\label{app:no.far}
The aim of this appendix is to extend the theory developed in Sect.~\ref{sect:time.dependent} to problems in which there is actually no far boundary within the host environment, that is the whole tissue is surrounded by blood vessels ($\partial\Omega\equiv\partial_b\Omega$, $\partial_f\Omega=\emptyset$). For instance, this situation arises when studying the confinement of a tumor mass within a healthy tissue well supplied with blood by a nearby vasculature. In this case we can imagine that $\overline{\Omega_T(t)}\subset\Omega$ at all times, with $\partial\Omega=\partial\Omega_H(t)\equiv\partial_b\Omega$ (Fig. \ref{fig:domain-nofar}). From the mathematical point of view, the main point is the disappearance of the Dirichlet boundary conditions, which forces one to modify the Poisson solution operator $\P$ introduced in Appendix~\ref{app:poisson} in order to deal with the new problem for $\phi$.

\begin{figure*}
\centering
\includegraphics[width=0.35\textwidth,clip]{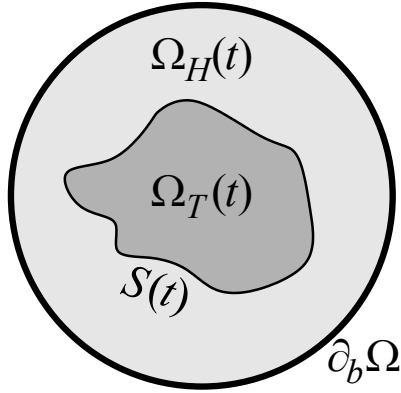}
\caption{The geometrical configuration of the domain $\Omega$ in the problem of a tumor mass surrounded by host tissue with no far boundary.}
\label{fig:domain-nofar}
\end{figure*}

Specifically, in place of Problem~\ref{eq:Poisson} we consider (cf. also \cite{fadimba1995ape})
\begin{equation*}
	\left\{
	\begin{array}{rcll}
		-\Delta{u} & = & f-\ave{f} & \text{in\ } \Omega \\[0.1cm]
		\nabla{u}\cdot\n & = & 0 & \text{on\ } \partial\Omega \\[0.1cm]
		\ave{u} & = & \ave{f} & \text{in\ } \Omega
	\end{array}
	\right.
\end{equation*}
for $f\in L^2(\Omega)$, where $\ave{\cdot}$ denotes average on $\Omega$:
\begin{equation*}
	\ave{f}:=\frac{1}{\vert\Omega\vert}\intl_\Omega f(x)\,dx.
\end{equation*}
The weak formulation is:
\begin{equation}
	\left\{
	\begin{array}{c}
		\text{find $u\in H^1(\Omega)$ such that} \\[0.2cm]
		\displaystyle{\intl_\Omega}\nabla{u}\cdot\nabla{v}\,dx+\vert\Omega\vert\ave{u}\ave{v}=
			\displaystyle{\intl_\Omega}fv\,dx, \quad \forall\,v\in H^1(\Omega),
	\end{array}
	\right.
	\label{eq:Poisson-Neumann-weak}
\end{equation}
whence, owing to Lax-Milgram theory, we get a unique weak solution fulfilling $\norm{u}_1\lesssim\norm{f}_0$. Setting $u=\P{f}$, we observe from Eq.~\ref{eq:Poisson-Neumann-weak} that the operator $\P$ satisfies now
\begin{equation*}
	\intl_\Omega\nabla{\P{f}}\cdot\nabla{v}\,dx+\vert\Omega\vert\ave{\P{f}}\ave{v}=\intl_\Omega fv\,dx, \quad
		\forall\,f\in L^2(\Omega),\,v\in H^1(\Omega),
\end{equation*}
therefore the inner product $\winnprod{\cdot}{\cdot}$ and the induced norm $\wnorm{\cdot}_0$ take the following forms:
\begin{align*}
	\winnprod{f}{g} &= \intl_\Omega\nabla{\P{f}}\cdot\nabla{\P{g}}\,dx+\vert\Omega\vert\ave{\P{f}}\ave{\P{g}} \\
	\wnorm{f}_0^2 &= \norm{\nabla{\P{f}}}_0^2+\vert\Omega\vert\ave{\P{f}}^2\asymp\norm{\P{f}}_1^2.
\end{align*}
This essentially affects the estimates of Theorem \ref{theo:uniqueness} when dealing with
\begin{multline*}
	\intl_\Omega\nabla{(\Phi(\phi_2)-\Phi(\phi_1))}\cdot\nabla{\P(\phi_2-\phi_1)}\,dx \\
	=\intl_\Omega(\Phi(\phi_2)-\Phi(\phi_1))(\phi_2-\phi_1)\,dx-
		\vert\Omega\vert\ave{\Phi(\phi_2)-\Phi(\phi_1)}\ave{\P(\phi_2-\phi_1)},
\end{multline*}
because of the second term at the right-hand side which must be estimated. First we use Cauchy's inequality to find
\begin{equation*}
	\ave{\Phi(\phi_2)-\Phi(\phi_1)}\ave{\P(\phi_2-\phi_1)}\leq
		\frac{\epsilon}{2}\ave{\Phi(\phi_2)-\Phi(\phi_1)}^2+\frac{1}{2\epsilon}\ave{\P(\phi_2-\phi_1)}^2.
\end{equation*}
Next we employ Cauchy-Schwartz's inequality and Lipschitz continuity of $\Phi$ to discover
\begin{align*}
	\ave{\Phi(\phi_2)-\Phi(\phi_1)}^2 &\leq
		\frac{1}{\vert\Omega\vert}\intl_\Omega\vert\Phi(\phi_2)-\Phi(\phi_1)\vert^2\,dx \\
	&\leq\frac{\Lip{\Phi}}{\vert\Omega\vert}
		\intl_\Omega(\Phi(\phi_2)-\Phi(\phi_1))(\phi_2-\phi_1)\,dx, \\
	\ave{\P(\phi_2-\phi_1)(t)}^2 &\leq \frac{1}{\vert\Omega\vert}\norm{\P(\phi_2-\phi_1)(t)}_0^2\leq
		C\wnorm{(\phi_2-\phi_1)(t)}_0^2,
\end{align*}
whence
\begin{multline*}
	\intl_\Omega\nabla{(\Phi(\phi_2)-\Phi(\phi_1))}\cdot\nabla{\P(\phi_2-\phi_1)}\,dx \\
		\geq\left(1-\frac{\epsilon\Lip{\Phi}}{2\vert\Omega\vert}\right)
			\intl_\Omega(\Phi(\phi_2)-\Phi(\phi_1))(\phi_2-\phi_1)\,dx-\frac{C}{2\epsilon}\wnorm{(\phi_2-\phi_1)(t)}_0^2.
\end{multline*}
Choosing $\epsilon>0$ so small that $\epsilon\Lip{\Phi}<2\abs{\Omega}$, the terms at the right-hand side can be finally incorporated into the similar ones already present in the proof of Theorem \ref{theo:uniqueness}.

\bibliographystyle{plain}
\bibliography{Ta-IBVtumor.bib}

\end{document}